\documentclass [12pt]{amsart}
\usepackage[utf8]{inputenc}
\pdfoutput=1
\usepackage{amsmath,amssymb,amsthm,amsfonts}
\usepackage{mathtools}

\DeclareRobustCommand{\gobblefour}[5]{}

\usepackage[english]{babel}
\usepackage{comment}
\usepackage{hyperref}
\usepackage{bbm}
\usepackage{bm}
\usepackage{mathrsfs}

\usepackage{tikz,graphicx,color}
\usepackage{tikz-cd}
\usepackage{tikz-3dplot}
\usetikzlibrary{calc}
\usetikzlibrary{arrows}
\usetikzlibrary{shapes}
\usetikzlibrary{patterns}
\usetikzlibrary{positioning}
\usetikzlibrary{arrows.meta}
\usetikzlibrary{decorations.markings}
\usepackage{xspace}

\pdfpageattr{/Group <</S /Transparency /I true /CS /DeviceRGB>>}

\usepackage{epstopdf}

\usepackage{enumerate}
\usepackage[arrow]{xy}
\usepackage{diagbox}
\usepackage{subfig}
\usepackage{arcs}
\usepackage{xcolor}
\usepackage{tabu}
\usepackage{booktabs}

\usepackage[margin=1in]{geometry}
\usepackage{lipsum}

\usepackage{enumitem}
\usepackage{letltxmacro}
\usepackage{thmtools,etoolbox}

\def\myarabic#1{\normalfont(\roman{#1})}
\newlist{theoremlist}{enumerate}{1}
\setlist[theoremlist]{label=\myarabic{theoremlisti},ref={\myarabic{theoremlisti}},itemindent=0pt,labelindent=0pt,
leftmargin=*,noitemsep}

\makeatletter
\renewcommand{\p@theoremlisti}{\perh@ps{\thetheorem}}
\protected\def\perh@ps#1#2{\textup{#1#2}}
\newcommand{\itemrefperh@ps}[2]{\textup{#2}}
\newcommand{\itemref}[1]{\begingroup\let\perh@ps\itemrefperh@ps\ref{#1}\endgroup}
\makeatother

\usepackage{nameref,hyperref}
\usepackage[capitalize]{cleveref}

\usepackage{autonum}
\newcommand{\retainlabel}[1]{\label{#1}\sbox0{\ref{#1}}}

\newtheorem{theorem}{Theorem}[section]

\newtheorem{lemma}[theorem]{Lemma}
\newtheorem{proposition}[theorem]{Proposition}
\newtheorem{corollary}[theorem]{Corollary}
\theoremstyle{definition}
\newtheorem{remark}[theorem]{Remark}
\theoremstyle{definition}
\newtheorem{definition}[theorem]{Definition}

\newtheorem{question}[theorem]{Question}
\theoremstyle{definition}

\theoremstyle{definition}
\newtheorem{example}[theorem]{Example}

\addtotheorempostheadhook[theorem]{\crefalias{theoremlisti}{theorem}}
\addtotheorempostheadhook[lemma]{\crefalias{theoremlisti}{lemma}}
\addtotheorempostheadhook[proposition]{\crefalias{theoremlisti}{proposition}}

\def\Acal{\mathcal{A}}

% ...

\def\one{{\mathbbm{1}}}
\def\C{\mathbb{C}}
\def\R{\mathbb{R}}

\def\Z{\mathbb{Z}}

% ...

\newcommand\parr[1]{{({#1})}}
% ...

\def\<{{\langle}}
\def\>{{\rangle}}

\def\diag{{ \operatorname{diag}}}

\def\op{{ \operatorname{op}}}

\def\Span{ \operatorname{Span}}

\def\Gr{\operatorname{Gr}}
\def\Grtnn{\Gr_{\ge 0}}

\def\xing{{\operatorname{xing}}}

\newcounter{todofigure}

\def\Prob{\operatorname{Prob}}

\def\OG{\operatorname{OG}}
\def\OGtnn{\OG_{\ge0}}

\def\doublemap{\phi}
\newcommand\double[1]{\widetilde{#1}}

\def\Reg{R}

\def\arg{\operatorname{arg}}

\def\curve{\bm{\gamma}}
\def\bth{{\bm{\theta}}}
\def\btht{{\tilde{\bm{\theta}}}}

\def\g{\gamma}
\def\J{J}

\def\bx{{\mathbf{x}}}
\def\Tiling{{\mathbb{T}}}

\crefname{figure}{Figure}{Figures}

\def\Arr{{\Acal}}
\def\ska#1{s_k\cdot #1}
\def\ska#1{#1\cdot s_k}
\def\coef{f}

\def\taut{{\tilde\tau}}
\def\Jt{{\tilde\J}}

\def\v{v}
\def\th{\theta}
\def\tht{{\tilde\theta}}
\def\vt{{\tilde\v}}
\def\eps{\varepsilon}
\def\bas{e}

\def\DUMMY{\Phi(t)}

\def\gg{\Gamma}
\def\GG{{\bm{\Gamma}}}
\def\supp{\operatorname{supp}_\tau}

\def\u#1{{u_\Reg^\parr{#1}}}

\def\RegN{\Reg_{\parr N}}

\def\bthN{\bth_{\parr N}}
\def\dd#1{\bar\partial_{#1}}
\def\Jp{J'}
\def\Jpp{J''}
\def\w#1{w^\parr{#1,N}}

\def\GGB{\GG^{\flat}_{\RegN}}
\def\GGBT{\GGB(T)}

\def\SpanC{\Span_\C}

\def\figref#1(#2){Figure~\hyperref[#1]{\ref*{#1}(#2)}}
\begin{document}
\numberwithin{equation}{section}

\title[A formula for boundary correlations of the critical Ising model]{A formula for boundary correlations \\ of the critical Ising model}
\author{Pavel Galashin}
\address{Department of Mathematics, University of California, Los Angeles, 520 Portola Plaza,
Los Angeles, CA 90025, USA}
\email{\href{mailto:galashin@math.ucla.edu}{galashin@math.ucla.edu}}

\date{\today}



\begin{abstract}
Given a finite rhombus tiling of a polygonal region in the plane, the associated critical $Z$-invariant Ising model is invariant under star-triangle transformations. We give a simple matrix formula describing spin correlations between boundary vertices in terms of the shape of the region. When the region is a regular polygon, our formula becomes an explicit trigonometric sum. 
\end{abstract}

\maketitle

\section{Introduction}\label{sec:intro}

Consider a rhombus tiling $\Tiling$ of a polygonal region $\Reg$ in the plane, such as the one in~\figref{fig:intro1}(a). To this data, one can attach a weighted \emph{isoradial graph} $G_\Tiling$ which depends on the geometry of the tiling in a simple local way. The associated \emph{critical $Z$-invariant Ising model} was introduced by Baxter~\cite{Bax,Bax2} and has been studied extensively since then; see e.g.~\cite{YangPerk, CS2, CS, CDC13, BdT1,BdT2,BdTR}. It is a probability measure on the space of spin configurations on the vertices of $G_\Tiling$, generalizing the Ising model at critical temperature on the square, triangular, and hexagonal lattices. Denote by $b_1,b_2,\dots,b_n$ the vertices of $G_\Tiling$ that belong to the boundary of $\Reg$, listed in counterclockwise order. For all $1\leq j,k\leq n$, let $\<\sigma_j\sigma_k\>_{\Reg}$ be the spin correlation between $b_j$ and $b_k$. By definition, $\<\sigma_j\sigma_k\>_{\Reg}$ is the difference between the probability that the spins at $b_j$ and $b_k$ are equal and the probability that these spins are different. See \cref{sec:ising-model} for details.

 It is known~\cite{KenyonAlg} that any two rhombus tilings of the same region can be related by a sequence of \emph{flips} as in \figref{fig:intro1}(d). Applying a flip to a rhombus tiling results in applying a \emph{star-triangle move} to the weighted graph $G_\Tiling$. The boundary correlations $\<\sigma_j\sigma_k\>_{\Reg}$ are preserved by such moves, and therefore depend only on the region $\Reg$ itself, and not on the particular choice of a rhombus tiling $\Tiling$. It is thus natural to look for an expression for $\<\sigma_j\sigma_k\>_{\Reg}$ purely in terms of $\Reg$. 
In this paper, building on our previous results with P.~Pylyavskyy~\cite{GP}, we give such an expression for an arbitrary region~$\Reg$.

\subsection{Regular polygons}\label{sec:regular-polygons-1}
In general, our formula involves computing the inverse of an $n\times n$ matrix. However, in the most symmetric case when $\Reg$ is a regular $2n$-gon, the matrix can be inverted explicitly, which gives rise to the following result.

\begin{theorem}\label{thm:reg}
  Let $\Reg_n$ be a regular $2n$-gon. Then for $1\leq p,q\leq n$ and $k:=|p-q|$, we have
\begin{equation}\label{eq:regular}
  \<\sigma_p\sigma_q\>_{\Reg_n}=\frac2n \left(\frac1{\sin \left((2k-1)\pi/2n\right)}-\frac1{\sin \left((2k-3)\pi/2n\right)}+\dots\pm \frac1{\sin \left(\pi/2n\right)} \right)\mp 1.
\end{equation}
\end{theorem}
\noindent Note that the right hand side of~\eqref{eq:regular} is invariant under the symmetry between $k$ and $n-k$.

 As we explain in \cref{rmk:KW}, this formula describes the unique $n\times n$ boundary correlation matrix of the Ising model that is invariant under the Kramers--Wannier duality~\cite{KrWa}. The proof of \cref{thm:reg} and its asymptotic consequences are presented in \cref{sec:regular-polygons}.

\begin{figure}
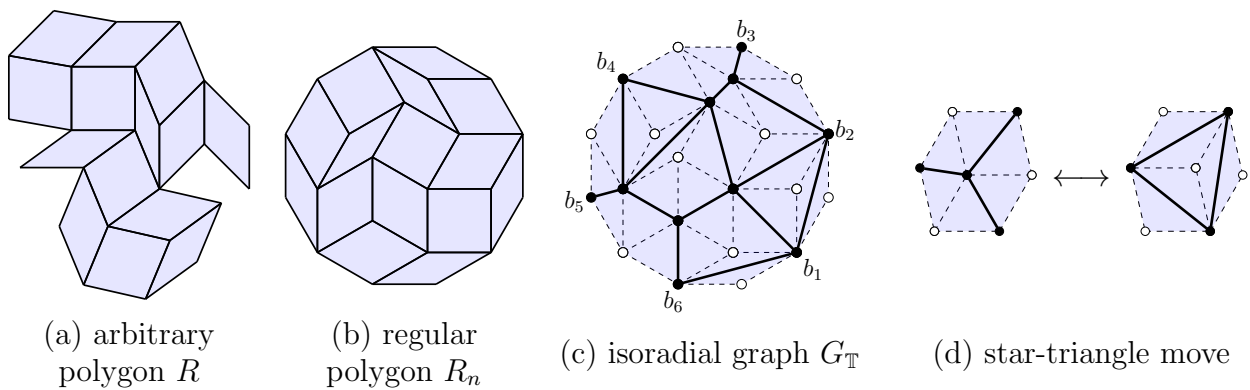
\makebox[1.0\textwidth]{

\setlength{\tabcolsep}{4pt}
} & \scalebox{1.00}{(c) isoradial graph $G_{\Tiling}$} & \scalebox{1.00}{(d) star-triangle move} 

\end{tabular}}

\caption{\label{fig:intro1} 
(a) A rhombus tiling of an arbitrary polygon $\Reg$; (b) a rhombus tiling of a regular polygon $\Reg_n$ for $n=6$; (c) the associated isoradial graph $G_\Tiling$ consists of black vertices and black solid edges; (d) a flip of a rhombus tiling resulting in a star-triangle move on $G_\Tiling$.}
\end{figure}

\subsection{Arbitrary regions}\label{sec:arbitrary-regions}
We now present our main result, \cref{thm:main_span} and \cref{cor:main}, which gives a formula for an arbitrary polygonal region $\Reg$. We identify vectors in the plane with complex numbers and denote $[2n]:=\{1,2,\dots,2n\}$.

Let $\Reg$ be a region whose boundary is a simple closed polygonal chain comprised of $2n$ unit vectors $\v_1,\v_2,\dots,v_{2n}\in\C$ listed and directed in the counterclockwise order.  Any rhombus tiling of $\Reg$ is dual to a \emph{pseudoline arrangement} obtained by connecting the midpoints of the opposite edges of each rhombus; see \figref{fig:intro2}(c). Each pseudoline connects the midpoints of $\v_j$ and $\v_k$ for some $j,k\in[2n]$, and we record the associated matching as a fixed-point-free involution $\tau=\tau_\Reg:[2n]\to[2n]$. By definition, we set $\tau(j):=k$ and $\tau(k):=j$ whenever the midpoints of $\v_j$ and $\v_k$ are connected by the same pseudoline. Clearly, we have 
\begin{equation}\label{eq:v=-v}
  \v_{\tau(j)}=-\v_j \quad\text{for all $j\in[2n]$.}
\end{equation}
Note that $\tau$ depends only on $\Reg$ and not on the choice of a rhombus tiling. For example, if $\Reg$ is given in \figref{fig:intro2}(a), then we have $n=9$ and $\tau:[2n]\to[2n]$ satisfies (cf. \figref{fig:intro2}(c))
\begin{equation}\label{eq:}
\tau(1)=7,\  \tau(2)=18,\  \tau(3)=12,\  \tau(4)=10,\ \tau(5)=8,\ \tau(6)=17,\ \dots,\ \tau(17)=6,\ \tau(18)=2.
\end{equation}

\begin{figure}
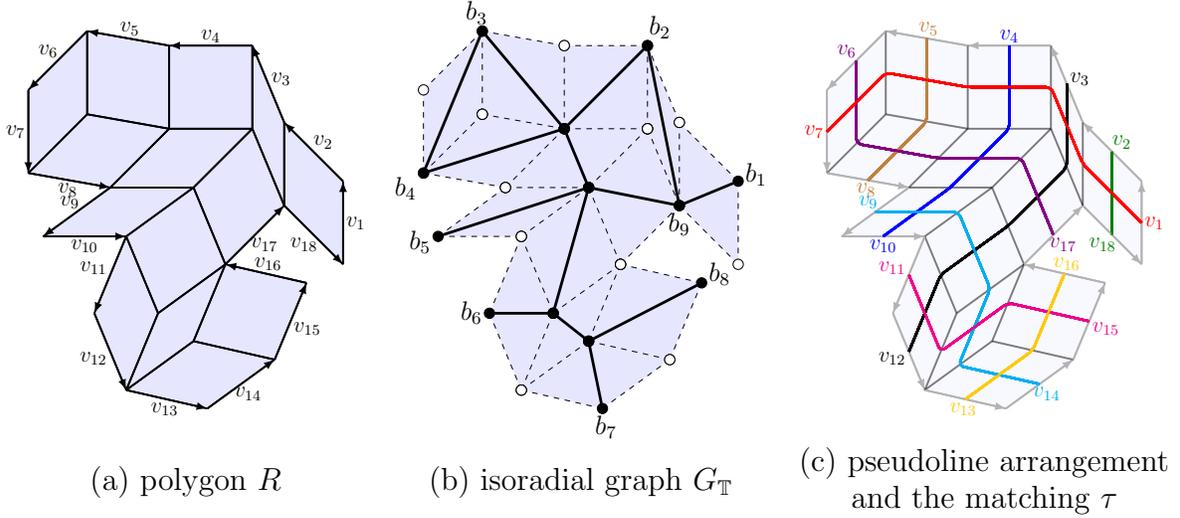


\makebox[1.0\textwidth]{
\setlength{\tabcolsep}{1pt}

}

\end{tabular}
}

\caption{\label{fig:intro2} The boundary vectors $\v_1,\dots,\v_{2n}$, the isoradial graph $G_\Tiling$, and the pseudoline arrangement (which determines the matching $\tau$) associated with a rhombus tiling $\Tiling$ of a polygonal region $\Reg$.}

\end{figure}

We would like to extract a square root of each $\v_j\in\C$ in a particular way. Namely, we choose angles $\th_1,\th_2,\dots,\th_{2n}\in\R$ satisfying 
\begin{align}
\label{eq:v=exp(th)}
\v_k&=\exp(2i\th_k), &&\text{for all $k\in[2n]$;}\\
\label{eq:th+pi/2}
\th_{\tau(k)}&=\th_k+\pi/2, &&\text{when $k<\tau(k)$;}\\
\label{eq:th<th<th<th}
\th_j<\th_k&<\th_{\tau(j)}<\th_{\tau(k)}, &&\text{when $j<k<\tau(j)<\tau(k)$.}
\end{align}
 For instance, if $\Reg$ is convex and $\arg(\v_1)\leq \arg(\v_k)$ for all $k\in[2n]$, we may choose $\th_k:=\arg(\v_k)/2$ for all $k\in[2n]$.

For the purposes of this introduction, we will impose the following ``non-alternating'' restriction on $\Reg$; see \cref{fig:alt}. It will be lifted later in \cref{sec:alternating-regions}.

\begin{definition}\label{dfn:non_alt}
\scalebox{1.0}{We say that $\Reg$ is \emph{alternating} if there exist $1\leq a<b<c<d\leq 2n$ such that}
\begin{equation}\label{eq:alternating}
\v_a=-\v_b=\v_c=-\v_d.
\end{equation}
Otherwise, we say that $R$ is \emph{non-alternating}.
\end{definition}
\noindent For example, in the \emph{generic} case where $\v_j\neq \v_k$ for all $j\neq k$, the region $\Reg$ is non-alternating.

\begin{definition}
For $a,b,c\in[2n]$, we say that $(a,b,c)$ \emph{form a counterclockwise triple} if either $a<b<c$, or $b<c<a$, or $c<a<b$. For a matching $\tau:[2n]\to[2n]$ and $k\in[2n]$, we set
\begin{equation}\label{eq:J_k}
\J_k:=\{j\in[2n]\mid (k,j,\tau(j))\text{ form a counterclockwise triple}\}.
\end{equation}
\end{definition}
\noindent Thus $\J_k$ is an $(n-1)$-element subset of $[2n]$ which contains exactly one element out of each pair $\{j,\tau(j)\}$ disjoint from $\{k,\tau(k)\}$.

\begin{figure}\makebox[1.0\textwidth]{

\setlength\arrayrulewidth{1pt}

\setlength{\tabcolsep}{4pt}
\begin{tabular}{c|c}\\
\scalebox{0.9}{
\begin{tabular}{ccc}

\scalebox{0.65}{
\begin{tikzpicture}[scale=1.15 ,baseline=(ZUZU.base)]\coordinate(ZUZU) at (0,0); 
\coordinate (X0) at (0.500,1.566);
\coordinate (X1) at (-0.500,1.566);
\coordinate (X2) at (-1.000,0.700);
\coordinate (X3) at (-0.000,0.700);
\coordinate (X4) at (0.000,2.432);
\coordinate (X5) at (-1.000,2.432);
\coordinate (X6) at (-0.000,-0.300);
\coordinate (X7) at (0.866,-0.800);
\coordinate (X8) at (0.866,0.200);
\coordinate (X9) at (-1.500,1.566);
\coordinate (X10) at (-1.000,-0.300);
\coordinate (X11) at (-2.366,1.066);
\coordinate (X12) at (-1.866,0.200);
\coordinate (X13) at (1.366,1.066);
\coordinate (X14) at (0.000,0.700);
\coordinate (X15) at (-1.866,-0.800);
\coordinate (X16) at (1.366,0.066);
\coordinate (X17) at (-0.866,-0.800);
\coordinate (X18) at (0.866,1.932);
\coordinate (X19) at (-0.000,-1.300);
\coordinate (X20) at (-1.866,1.932);
\coordinate (X21) at (-2.366,0.066);
\coordinate (X22) at (-1.000,-1.300);
\coordinate (X23) at (0.433,-1.050);
\coordinate (X24) at (1.116,-0.367);
\coordinate (X25) at (1.366,0.566);
\coordinate (X26) at (1.116,1.499);
\coordinate (X27) at (0.433,2.182);
\coordinate (X28) at (-0.500,2.432);
\coordinate (X29) at (-1.433,2.182);
\coordinate (X30) at (-2.116,1.499);
\coordinate (X31) at (-2.366,0.566);
\coordinate (X32) at (-2.116,-0.367);
\coordinate (X33) at (-1.433,-1.050);
\coordinate (X34) at (-0.500,-1.300);
\fill [opacity=0.10,blue] (X0)-- (X1)-- (X2)-- (X3) -- cycle;
\draw [line width=1pt,black,opacity=1.00] (X0)-- (X1);
\draw [line width=1pt,black,opacity=1.00] (X1)-- (X2);
\draw [line width=1pt,black,opacity=1.00] (X2)-- (X3);
\draw [line width=1pt,black,opacity=1.00] (X3)-- (X0);
\fill [opacity=0.10,blue] (X4)-- (X5)-- (X1)-- (X0) -- cycle;
\draw [line width=1pt,black,opacity=1.00] (X4)-- (X5);
\draw [line width=1pt,black,opacity=1.00] (X5)-- (X1);
\draw [line width=1pt,black,opacity=1.00] (X1)-- (X0);
\draw [line width=1pt,black,opacity=1.00] (X0)-- (X4);
\fill [opacity=0.10,blue] (X3)-- (X6)-- (X7)-- (X8) -- cycle;
\draw [line width=1pt,black,opacity=1.00] (X3)-- (X6);
\draw [line width=1pt,black,opacity=1.00] (X6)-- (X7);
\draw [line width=1pt,black,opacity=1.00] (X7)-- (X8);
\draw [line width=1pt,black,opacity=1.00] (X8)-- (X3);
\fill [opacity=0.10,blue] (X5)-- (X9)-- (X2)-- (X1) -- cycle;
\draw [line width=1pt,black,opacity=1.00] (X5)-- (X9);
\draw [line width=1pt,black,opacity=1.00] (X9)-- (X2);
\draw [line width=1pt,black,opacity=1.00] (X2)-- (X1);
\draw [line width=1pt,black,opacity=1.00] (X1)-- (X5);
\fill [opacity=0.10,blue] (X2)-- (X10)-- (X6)-- (X3) -- cycle;
\draw [line width=1pt,black,opacity=1.00] (X2)-- (X10);
\draw [line width=1pt,black,opacity=1.00] (X10)-- (X6);
\draw [line width=1pt,black,opacity=1.00] (X6)-- (X3);
\draw [line width=1pt,black,opacity=1.00] (X3)-- (X2);
\fill [opacity=0.10,blue] (X9)-- (X11)-- (X12)-- (X2) -- cycle;
\draw [line width=1pt,black,opacity=1.00] (X9)-- (X11);
\draw [line width=1pt,black,opacity=1.00] (X11)-- (X12);
\draw [line width=1pt,black,opacity=1.00] (X12)-- (X2);
\draw [line width=1pt,black,opacity=1.00] (X2)-- (X9);
\fill [opacity=0.10,blue] (X8)-- (X13)-- (X0)-- (X14) -- cycle;
\draw [line width=1pt,black,opacity=1.00] (X8)-- (X13);
\draw [line width=1pt,black,opacity=1.00] (X13)-- (X0);
\draw [line width=1pt,black,opacity=1.00] (X0)-- (X14);
\draw [line width=1pt,black,opacity=1.00] (X14)-- (X8);
\fill [opacity=0.10,blue] (X12)-- (X15)-- (X10)-- (X2) -- cycle;
\draw [line width=1pt,black,opacity=1.00] (X12)-- (X15);
\draw [line width=1pt,black,opacity=1.00] (X15)-- (X10);
\draw [line width=1pt,black,opacity=1.00] (X10)-- (X2);
\draw [line width=1pt,black,opacity=1.00] (X2)-- (X12);
\fill [opacity=0.10,blue] (X7)-- (X16)-- (X13)-- (X8) -- cycle;
\draw [line width=1pt,black,opacity=1.00] (X7)-- (X16);
\draw [line width=1pt,black,opacity=1.00] (X16)-- (X13);
\draw [line width=1pt,black,opacity=1.00] (X13)-- (X8);
\draw [line width=1pt,black,opacity=1.00] (X8)-- (X7);
\fill [opacity=0.10,blue] (X15)-- (X17)-- (X6)-- (X10) -- cycle;
\draw [line width=1pt,black,opacity=1.00] (X15)-- (X17);
\draw [line width=1pt,black,opacity=1.00] (X17)-- (X6);
\draw [line width=1pt,black,opacity=1.00] (X6)-- (X10);
\draw [line width=1pt,black,opacity=1.00] (X10)-- (X15);
\fill [opacity=0.10,blue] (X13)-- (X18)-- (X4)-- (X0) -- cycle;
\draw [line width=1pt,black,opacity=1.00] (X13)-- (X18);
\draw [line width=1pt,black,opacity=1.00] (X18)-- (X4);
\draw [line width=1pt,black,opacity=1.00] (X4)-- (X0);
\draw [line width=1pt,black,opacity=1.00] (X0)-- (X13);
\fill [opacity=0.10,blue] (X17)-- (X19)-- (X7)-- (X6) -- cycle;
\draw [line width=1pt,black,opacity=1.00] (X17)-- (X19);
\draw [line width=1pt,black,opacity=1.00] (X19)-- (X7);
\draw [line width=1pt,black,opacity=1.00] (X7)-- (X6);
\draw [line width=1pt,black,opacity=1.00] (X6)-- (X17);
\fill [opacity=0.10,blue] (X5)-- (X20)-- (X11)-- (X9) -- cycle;
\draw [line width=1pt,black,opacity=1.00] (X5)-- (X20);
\draw [line width=1pt,black,opacity=1.00] (X20)-- (X11);
\draw [line width=1pt,black,opacity=1.00] (X11)-- (X9);
\draw [line width=1pt,black,opacity=1.00] (X9)-- (X5);
\fill [opacity=0.10,blue] (X11)-- (X21)-- (X15)-- (X12) -- cycle;
\draw [line width=1pt,black,opacity=1.00] (X11)-- (X21);
\draw [line width=1pt,black,opacity=1.00] (X21)-- (X15);
\draw [line width=1pt,black,opacity=1.00] (X15)-- (X12);
\draw [line width=1pt,black,opacity=1.00] (X12)-- (X11);
\fill [opacity=0.10,blue] (X15)-- (X22)-- (X19)-- (X17) -- cycle;
\draw [line width=1pt,black,opacity=1.00] (X15)-- (X22);
\draw [line width=1pt,black,opacity=1.00] (X22)-- (X19);
\draw [line width=1pt,black,opacity=1.00] (X19)-- (X17);
\draw [line width=1pt,black,opacity=1.00] (X17)-- (X15);
\end{tikzpicture}
}
 & \scalebox{0.65}{
\begin{tikzpicture}[scale=1.15 ,baseline=(ZUZU.base)]\coordinate(ZUZU) at (0,0); 
\coordinate (X0) at (-0.223,0.975);
\coordinate (X1) at (-0.223,1.975);
\coordinate (X2) at (-1.223,1.975);
\coordinate (X3) at (-1.223,0.975);
\coordinate (X4) at (0.777,0.975);
\coordinate (X5) at (0.777,1.975);
\coordinate (X6) at (0.000,0.000);
\coordinate (X7) at (1.000,0.000);
\coordinate (X8) at (-2.223,0.975);
\coordinate (X9) at (-2.223,-0.025);
\coordinate (X10) at (-1.223,-0.025);
\coordinate (X11) at (-2.223,1.975);
\coordinate (X12) at (0.500,0.000);
\coordinate (X13) at (0.889,0.487);
\coordinate (X14) at (0.777,1.475);
\coordinate (X15) at (0.277,1.975);
\coordinate (X16) at (-0.723,1.975);
\coordinate (X17) at (-1.723,1.975);
\coordinate (X18) at (-2.223,1.475);
\coordinate (X19) at (-2.223,0.475);
\coordinate (X20) at (-1.723,-0.025);
\coordinate (X21) at (-1.223,0.475);
\coordinate (X22) at (-0.723,0.975);
\coordinate (X23) at (-0.111,0.487);
\fill [opacity=0.10,blue] (X0)-- (X1)-- (X2)-- (X3) -- cycle;
\draw [line width=1pt,black,opacity=1.00] (X0)-- (X1);
\draw [line width=1pt,black,opacity=1.00] (X1)-- (X2);
\draw [line width=1pt,black,opacity=1.00] (X2)-- (X3);
\draw [line width=1pt,black,opacity=1.00] (X3)-- (X0);
\fill [opacity=0.10,blue] (X4)-- (X5)-- (X1)-- (X0) -- cycle;
\draw [line width=1pt,black,opacity=1.00] (X4)-- (X5);
\draw [line width=1pt,black,opacity=1.00] (X5)-- (X1);
\draw [line width=1pt,black,opacity=1.00] (X1)-- (X0);
\draw [line width=1pt,black,opacity=1.00] (X0)-- (X4);
\fill [opacity=0.10,blue] (X0)-- (X6)-- (X7)-- (X4) -- cycle;
\draw [line width=1pt,black,opacity=1.00] (X0)-- (X6);
\draw [line width=1pt,black,opacity=1.00] (X6)-- (X7);
\draw [line width=1pt,black,opacity=1.00] (X7)-- (X4);
\draw [line width=1pt,black,opacity=1.00] (X4)-- (X0);
\fill [opacity=0.10,blue] (X8)-- (X9)-- (X10)-- (X3) -- cycle;
\draw [line width=1pt,black,opacity=1.00] (X8)-- (X9);
\draw [line width=1pt,black,opacity=1.00] (X9)-- (X10);
\draw [line width=1pt,black,opacity=1.00] (X10)-- (X3);
\draw [line width=1pt,black,opacity=1.00] (X3)-- (X8);
\fill [opacity=0.10,blue] (X2)-- (X11)-- (X8)-- (X3) -- cycle;
\draw [line width=1pt,black,opacity=1.00] (X2)-- (X11);
\draw [line width=1pt,black,opacity=1.00] (X11)-- (X8);
\draw [line width=1pt,black,opacity=1.00] (X8)-- (X3);
\draw [line width=1pt,black,opacity=1.00] (X3)-- (X2);
\end{tikzpicture}
}
 & \scalebox{0.65}{
\begin{tikzpicture}[scale=1.15 ,baseline=(ZUZU.base)]\coordinate(ZUZU) at (0,0); 
\coordinate (X0) at (-1.090,2.631);
\coordinate (X1) at (-1.797,1.924);
\coordinate (X2) at (-1.414,1.000);
\coordinate (X3) at (-0.707,1.707);
\coordinate (X4) at (-2.090,2.631);
\coordinate (X5) at (-2.797,1.924);
\coordinate (X6) at (-1.797,0.924);
\coordinate (X7) at (-1.414,0.000);
\coordinate (X8) at (-3.075,2.805);
\coordinate (X9) at (-3.782,2.098);
\coordinate (X10) at (-0.707,0.707);
\coordinate (X11) at (-2.606,0.336);
\coordinate (X12) at (-2.223,-0.588);
\coordinate (X13) at (-2.606,-1.512);
\coordinate (X14) at (-1.631,-1.734);
\coordinate (X15) at (-1.248,-0.810);
\coordinate (X16) at (-2.797,0.924);
\coordinate (X17) at (-3.782,1.098);
\coordinate (X18) at (-3.606,0.336);
\coordinate (X19) at (0.000,0.000);
\coordinate (X20) at (0.000,1.000);
\coordinate (X21) at (-0.439,-0.223);
\coordinate (X22) at (-0.822,-1.146);
\coordinate (X23) at (-2.989,-0.588);
\coordinate (X24) at (0.000,0.500);
\coordinate (X25) at (-0.354,1.354);
\coordinate (X26) at (-0.898,2.169);
\coordinate (X27) at (-1.590,2.631);
\coordinate (X28) at (-2.582,2.718);
\coordinate (X29) at (-3.428,2.451);
\coordinate (X30) at (-3.782,1.598);
\coordinate (X31) at (-3.289,1.011);
\coordinate (X32) at (-3.201,0.630);
\coordinate (X33) at (-3.106,0.336);
\coordinate (X34) at (-2.797,-0.126);
\coordinate (X35) at (-2.797,-1.050);
\coordinate (X36) at (-2.118,-1.623);
\coordinate (X37) at (-1.226,-1.440);
\coordinate (X38) at (-0.631,-0.684);
\coordinate (X39) at (-0.927,-0.111);
\coordinate (X40) at (-1.061,0.354);
\coordinate (X41) at (-0.354,0.354);
\fill [opacity=0.10,blue] (X0)-- (X1)-- (X2)-- (X3) -- cycle;
\draw [line width=1pt,black,opacity=1.00] (X0)-- (X1);
\draw [line width=1pt,black,opacity=1.00] (X1)-- (X2);
\draw [line width=1pt,black,opacity=1.00] (X2)-- (X3);
\draw [line width=1pt,black,opacity=1.00] (X3)-- (X0);
\fill [opacity=0.10,blue] (X4)-- (X5)-- (X1)-- (X0) -- cycle;
\draw [line width=1pt,black,opacity=1.00] (X4)-- (X5);
\draw [line width=1pt,black,opacity=1.00] (X5)-- (X1);
\draw [line width=1pt,black,opacity=1.00] (X1)-- (X0);
\draw [line width=1pt,black,opacity=1.00] (X0)-- (X4);
\fill [opacity=0.10,blue] (X1)-- (X6)-- (X7)-- (X2) -- cycle;
\draw [line width=1pt,black,opacity=1.00] (X1)-- (X6);
\draw [line width=1pt,black,opacity=1.00] (X6)-- (X7);
\draw [line width=1pt,black,opacity=1.00] (X7)-- (X2);
\draw [line width=1pt,black,opacity=1.00] (X2)-- (X1);
\fill [opacity=0.10,blue] (X8)-- (X9)-- (X5)-- (X4) -- cycle;
\draw [line width=1pt,black,opacity=1.00] (X8)-- (X9);
\draw [line width=1pt,black,opacity=1.00] (X9)-- (X5);
\draw [line width=1pt,black,opacity=1.00] (X5)-- (X4);
\draw [line width=1pt,black,opacity=1.00] (X4)-- (X8);
\fill [opacity=0.10,blue] (X7)-- (X10)-- (X3)-- (X2) -- cycle;
\draw [line width=1pt,black,opacity=1.00] (X7)-- (X10);
\draw [line width=1pt,black,opacity=1.00] (X10)-- (X3);
\draw [line width=1pt,black,opacity=1.00] (X3)-- (X2);
\draw [line width=1pt,black,opacity=1.00] (X2)-- (X7);
\fill [opacity=0.10,blue] (X6)-- (X11)-- (X12)-- (X7) -- cycle;
\draw [line width=1pt,black,opacity=1.00] (X6)-- (X11);
\draw [line width=1pt,black,opacity=1.00] (X11)-- (X12);
\draw [line width=1pt,black,opacity=1.00] (X12)-- (X7);
\draw [line width=1pt,black,opacity=1.00] (X7)-- (X6);
\fill [opacity=0.10,blue] (X13)-- (X14)-- (X15)-- (X12) -- cycle;
\draw [line width=1pt,black,opacity=1.00] (X13)-- (X14);
\draw [line width=1pt,black,opacity=1.00] (X14)-- (X15);
\draw [line width=1pt,black,opacity=1.00] (X15)-- (X12);
\draw [line width=1pt,black,opacity=1.00] (X12)-- (X13);
\fill [opacity=0.10,blue] (X5)-- (X16)-- (X6)-- (X1) -- cycle;
\draw [line width=1pt,black,opacity=1.00] (X5)-- (X16);
\draw [line width=1pt,black,opacity=1.00] (X16)-- (X6);
\draw [line width=1pt,black,opacity=1.00] (X6)-- (X1);
\draw [line width=1pt,black,opacity=1.00] (X1)-- (X5);
\fill [opacity=0.10,blue] (X9)-- (X17)-- (X16)-- (X5) -- cycle;
\draw [line width=1pt,black,opacity=1.00] (X9)-- (X17);
\draw [line width=1pt,black,opacity=1.00] (X17)-- (X16);
\draw [line width=1pt,black,opacity=1.00] (X16)-- (X5);
\draw [line width=1pt,black,opacity=1.00] (X5)-- (X9);
\fill [opacity=0.10,blue] (X16)-- (X18)-- (X11)-- (X6) -- cycle;
\draw [line width=1pt,black,opacity=1.00] (X16)-- (X18);
\draw [line width=1pt,black,opacity=1.00] (X18)-- (X11);
\draw [line width=1pt,black,opacity=1.00] (X11)-- (X6);
\draw [line width=1pt,black,opacity=1.00] (X6)-- (X16);
\fill [opacity=0.10,blue] (X10)-- (X19)-- (X20)-- (X3) -- cycle;
\draw [line width=1pt,black,opacity=1.00] (X10)-- (X19);
\draw [line width=1pt,black,opacity=1.00] (X19)-- (X20);
\draw [line width=1pt,black,opacity=1.00] (X20)-- (X3);
\draw [line width=1pt,black,opacity=1.00] (X3)-- (X10);
\fill [opacity=0.10,blue] (X15)-- (X21)-- (X7)-- (X12) -- cycle;
\draw [line width=1pt,black,opacity=1.00] (X15)-- (X21);
\draw [line width=1pt,black,opacity=1.00] (X21)-- (X7);
\draw [line width=1pt,black,opacity=1.00] (X7)-- (X12);
\draw [line width=1pt,black,opacity=1.00] (X12)-- (X15);
\fill [opacity=0.10,blue] (X14)-- (X22)-- (X21)-- (X15) -- cycle;
\draw [line width=1pt,black,opacity=1.00] (X14)-- (X22);
\draw [line width=1pt,black,opacity=1.00] (X22)-- (X21);
\draw [line width=1pt,black,opacity=1.00] (X21)-- (X15);
\draw [line width=1pt,black,opacity=1.00] (X15)-- (X14);
\fill [opacity=0.10,blue] (X11)-- (X23)-- (X13)-- (X12) -- cycle;
\draw [line width=1pt,black,opacity=1.00] (X11)-- (X23);
\draw [line width=1pt,black,opacity=1.00] (X23)-- (X13);
\draw [line width=1pt,black,opacity=1.00] (X13)-- (X12);
\draw [line width=1pt,black,opacity=1.00] (X12)-- (X11);
\end{tikzpicture}
}

\end{tabular} 
}
& 
\scalebox{0.9}{
\begin{tabular}{cc}

\scalebox{0.65}{
\begin{tikzpicture}[scale=1.15 ,baseline=(ZUZU.base)]\coordinate(ZUZU) at (0,0); 
\coordinate (X0) at (0.000,1.200);
\coordinate (X1) at (-1.000,1.200);
\coordinate (X2) at (-1.000,0.200);
\coordinate (X3) at (0.000,0.200);
\coordinate (X4) at (-2.000,1.200);
\coordinate (X5) at (-2.000,0.200);
\coordinate (X6) at (1.000,0.200);
\coordinate (X7) at (1.000,1.200);
\coordinate (X8) at (-2.000,-0.800);
\coordinate (X9) at (-1.000,-0.800);
\coordinate (X10) at (-0.000,-0.800);
\coordinate (X11) at (1.000,-0.800);
\coordinate (X12) at (0.500,-0.800);
\coordinate (X13) at (1.000,-0.300);
\coordinate (X14) at (1.000,0.700);
\coordinate (X15) at (0.500,1.200);
\coordinate (X16) at (-0.500,1.200);
\coordinate (X17) at (-1.500,1.200);
\coordinate (X18) at (-2.000,0.700);
\coordinate (X19) at (-2.000,-0.300);
\coordinate (X20) at (-1.500,-0.800);
\coordinate (X21) at (-1.000,-0.300);
\coordinate (X22) at (-0.500,0.200);
\coordinate (X23) at (0.000,-0.800);
\coordinate (X24) at (0.000,-0.300);
\fill [opacity=0.10,blue] (X0)-- (X1)-- (X2)-- (X3) -- cycle;
\draw [line width=1pt,black,opacity=1.00] (X0)-- (X1);
\draw [line width=1pt,black,opacity=1.00] (X1)-- (X2);
\draw [line width=1pt,black,opacity=1.00] (X2)-- (X3);
\draw [line width=1pt,black,opacity=1.00] (X3)-- (X0);
\fill [opacity=0.10,blue] (X1)-- (X4)-- (X5)-- (X2) -- cycle;
\draw [line width=1pt,black,opacity=1.00] (X1)-- (X4);
\draw [line width=1pt,black,opacity=1.00] (X4)-- (X5);
\draw [line width=1pt,black,opacity=1.00] (X5)-- (X2);
\draw [line width=1pt,black,opacity=1.00] (X2)-- (X1);
\fill [opacity=0.10,blue] (X6)-- (X7)-- (X0)-- (X3) -- cycle;
\draw [line width=1pt,black,opacity=1.00] (X6)-- (X7);
\draw [line width=1pt,black,opacity=1.00] (X7)-- (X0);
\draw [line width=1pt,black,opacity=1.00] (X0)-- (X3);
\draw [line width=1pt,black,opacity=1.00] (X3)-- (X6);
\fill [opacity=0.10,blue] (X5)-- (X8)-- (X9)-- (X2) -- cycle;
\draw [line width=1pt,black,opacity=1.00] (X5)-- (X8);
\draw [line width=1pt,black,opacity=1.00] (X8)-- (X9);
\draw [line width=1pt,black,opacity=1.00] (X9)-- (X2);
\draw [line width=1pt,black,opacity=1.00] (X2)-- (X5);
\fill [opacity=0.10,blue] (X10)-- (X11)-- (X6)-- (X3) -- cycle;
\draw [line width=1pt,black,opacity=1.00] (X10)-- (X11);
\draw [line width=1pt,black,opacity=1.00] (X11)-- (X6);
\draw [line width=1pt,black,opacity=1.00] (X6)-- (X3);
\draw [line width=1pt,black,opacity=1.00] (X3)-- (X10);
\node[inner sep=1.5pt,scale=1,anchor=180,color=red] (Vd) at (X13) {$\v_{d}$};
\draw[->,line width=2pt,red,>={latex}] (X11) -- (X6);
\node[inner sep=1.5pt,scale=1,anchor=360,color=red] (Va) at (X19) {$\v_{a}$};
\draw[->,line width=2pt,red,>={latex}] (X5) -- (X8);
\node[inner sep=1.5pt,scale=1,anchor=540,color=red] (Vb) at (X21) {$\v_{b}$};
\draw[->,line width=2pt,red,>={latex}] (X9) -- (X2);
\node[inner sep=1.5pt,scale=1,anchor=360,color=red] (Vc) at (X24) {$\v_{c}$};
\draw[->,line width=2pt,red,>={latex}] (X3) -- (X23);
\end{tikzpicture}
}
 & \scalebox{0.65}{
\begin{tikzpicture}[scale=1.15 ,baseline=(ZUZU.base)]\coordinate(ZUZU) at (0,0); 
\coordinate (X0) at (-1.090,2.631);
\coordinate (X1) at (-1.797,1.924);
\coordinate (X2) at (-1.414,1.000);
\coordinate (X3) at (-0.707,1.707);
\coordinate (X4) at (-2.090,2.631);
\coordinate (X5) at (-2.797,1.924);
\coordinate (X6) at (-1.797,0.924);
\coordinate (X7) at (-1.414,0.000);
\coordinate (X8) at (-3.075,2.805);
\coordinate (X9) at (-3.782,2.098);
\coordinate (X10) at (-0.707,0.707);
\coordinate (X11) at (-2.606,0.336);
\coordinate (X12) at (-2.223,-0.588);
\coordinate (X13) at (-2.606,-1.512);
\coordinate (X14) at (-1.899,-2.219);
\coordinate (X15) at (-1.516,-1.295);
\coordinate (X16) at (-2.797,0.924);
\coordinate (X17) at (-3.782,1.098);
\coordinate (X18) at (-3.606,0.336);
\coordinate (X19) at (0.000,0.000);
\coordinate (X20) at (0.000,1.000);
\coordinate (X21) at (-0.707,-0.707);
\coordinate (X22) at (-1.090,-1.631);
\coordinate (X23) at (-2.989,-0.588);
\coordinate (X24) at (0.000,0.500);
\coordinate (X25) at (-0.354,1.354);
\coordinate (X26) at (-0.898,2.169);
\coordinate (X27) at (-1.590,2.631);
\coordinate (X28) at (-2.582,2.718);
\coordinate (X29) at (-3.428,2.451);
\coordinate (X30) at (-3.782,1.598);
\coordinate (X31) at (-3.289,1.011);
\coordinate (X32) at (-3.201,0.630);
\coordinate (X33) at (-3.106,0.336);
\coordinate (X34) at (-2.797,-0.126);
\coordinate (X35) at (-2.797,-1.050);
\coordinate (X36) at (-2.252,-1.865);
\coordinate (X37) at (-1.494,-1.925);
\coordinate (X38) at (-0.898,-1.169);
\coordinate (X39) at (-1.061,-0.354);
\coordinate (X40) at (-1.061,0.354);
\coordinate (X41) at (-0.354,0.354);
\fill [opacity=0.10,blue] (X0)-- (X1)-- (X2)-- (X3) -- cycle;
\draw [line width=1pt,black,opacity=1.00] (X0)-- (X1);
\draw [line width=1pt,black,opacity=1.00] (X1)-- (X2);
\draw [line width=1pt,black,opacity=1.00] (X2)-- (X3);
\draw [line width=1pt,black,opacity=1.00] (X3)-- (X0);
\fill [opacity=0.10,blue] (X4)-- (X5)-- (X1)-- (X0) -- cycle;
\draw [line width=1pt,black,opacity=1.00] (X4)-- (X5);
\draw [line width=1pt,black,opacity=1.00] (X5)-- (X1);
\draw [line width=1pt,black,opacity=1.00] (X1)-- (X0);
\draw [line width=1pt,black,opacity=1.00] (X0)-- (X4);
\fill [opacity=0.10,blue] (X1)-- (X6)-- (X7)-- (X2) -- cycle;
\draw [line width=1pt,black,opacity=1.00] (X1)-- (X6);
\draw [line width=1pt,black,opacity=1.00] (X6)-- (X7);
\draw [line width=1pt,black,opacity=1.00] (X7)-- (X2);
\draw [line width=1pt,black,opacity=1.00] (X2)-- (X1);
\fill [opacity=0.10,blue] (X8)-- (X9)-- (X5)-- (X4) -- cycle;
\draw [line width=1pt,black,opacity=1.00] (X8)-- (X9);
\draw [line width=1pt,black,opacity=1.00] (X9)-- (X5);
\draw [line width=1pt,black,opacity=1.00] (X5)-- (X4);
\draw [line width=1pt,black,opacity=1.00] (X4)-- (X8);
\fill [opacity=0.10,blue] (X7)-- (X10)-- (X3)-- (X2) -- cycle;
\draw [line width=1pt,black,opacity=1.00] (X7)-- (X10);
\draw [line width=1pt,black,opacity=1.00] (X10)-- (X3);
\draw [line width=1pt,black,opacity=1.00] (X3)-- (X2);
\draw [line width=1pt,black,opacity=1.00] (X2)-- (X7);
\fill [opacity=0.10,blue] (X6)-- (X11)-- (X12)-- (X7) -- cycle;
\draw [line width=1pt,black,opacity=1.00] (X6)-- (X11);
\draw [line width=1pt,black,opacity=1.00] (X11)-- (X12);
\draw [line width=1pt,black,opacity=1.00] (X12)-- (X7);
\draw [line width=1pt,black,opacity=1.00] (X7)-- (X6);
\fill [opacity=0.10,blue] (X13)-- (X14)-- (X15)-- (X12) -- cycle;
\draw [line width=1pt,black,opacity=1.00] (X13)-- (X14);
\draw [line width=1pt,black,opacity=1.00] (X14)-- (X15);
\draw [line width=1pt,black,opacity=1.00] (X15)-- (X12);
\draw [line width=1pt,black,opacity=1.00] (X12)-- (X13);
\fill [opacity=0.10,blue] (X5)-- (X16)-- (X6)-- (X1) -- cycle;
\draw [line width=1pt,black,opacity=1.00] (X5)-- (X16);
\draw [line width=1pt,black,opacity=1.00] (X16)-- (X6);
\draw [line width=1pt,black,opacity=1.00] (X6)-- (X1);
\draw [line width=1pt,black,opacity=1.00] (X1)-- (X5);
\fill [opacity=0.10,blue] (X9)-- (X17)-- (X16)-- (X5) -- cycle;
\draw [line width=1pt,black,opacity=1.00] (X9)-- (X17);
\draw [line width=1pt,black,opacity=1.00] (X17)-- (X16);
\draw [line width=1pt,black,opacity=1.00] (X16)-- (X5);
\draw [line width=1pt,black,opacity=1.00] (X5)-- (X9);
\fill [opacity=0.10,blue] (X16)-- (X18)-- (X11)-- (X6) -- cycle;
\draw [line width=1pt,black,opacity=1.00] (X16)-- (X18);
\draw [line width=1pt,black,opacity=1.00] (X18)-- (X11);
\draw [line width=1pt,black,opacity=1.00] (X11)-- (X6);
\draw [line width=1pt,black,opacity=1.00] (X6)-- (X16);
\fill [opacity=0.10,blue] (X10)-- (X19)-- (X20)-- (X3) -- cycle;
\draw [line width=1pt,black,opacity=1.00] (X10)-- (X19);
\draw [line width=1pt,black,opacity=1.00] (X19)-- (X20);
\draw [line width=1pt,black,opacity=1.00] (X20)-- (X3);
\draw [line width=1pt,black,opacity=1.00] (X3)-- (X10);
\fill [opacity=0.10,blue] (X15)-- (X21)-- (X7)-- (X12) -- cycle;
\draw [line width=1pt,black,opacity=1.00] (X15)-- (X21);
\draw [line width=1pt,black,opacity=1.00] (X21)-- (X7);
\draw [line width=1pt,black,opacity=1.00] (X7)-- (X12);
\draw [line width=1pt,black,opacity=1.00] (X12)-- (X15);
\fill [opacity=0.10,blue] (X14)-- (X22)-- (X21)-- (X15) -- cycle;
\draw [line width=1pt,black,opacity=1.00] (X14)-- (X22);
\draw [line width=1pt,black,opacity=1.00] (X22)-- (X21);
\draw [line width=1pt,black,opacity=1.00] (X21)-- (X15);
\draw [line width=1pt,black,opacity=1.00] (X15)-- (X14);
\fill [opacity=0.10,blue] (X11)-- (X23)-- (X13)-- (X12) -- cycle;
\draw [line width=1pt,black,opacity=1.00] (X11)-- (X23);
\draw [line width=1pt,black,opacity=1.00] (X23)-- (X13);
\draw [line width=1pt,black,opacity=1.00] (X13)-- (X12);
\draw [line width=1pt,black,opacity=1.00] (X12)-- (X11);
\node[inner sep=1.5pt,scale=1,anchor=225,color=red] (Vd) at (X25) {$\v_{d}$};
\draw[->,line width=2pt,red,>={latex}] (X20) -- (X3);
\node[inner sep=1.5pt,scale=1,anchor=405,color=red] (Va) at (X36) {$\v_{a}$};
\draw[->,line width=2pt,red,>={latex}] (X13) -- (X14);
\node[inner sep=1.5pt,scale=1,anchor=585,color=red] (Vb) at (X39) {$\v_{b}$};
\draw[->,line width=2pt,red,>={latex}] (X21) -- (X7);
\node[inner sep=1.5pt,scale=1,anchor=405,color=red] (Vc) at (X41) {$\v_{c}$};
\draw[->,line width=2pt,red,>={latex}] (X10) -- (X19);
\end{tikzpicture}
}

\end{tabular}
}
 \\

 & \\

\scalebox{1}{non-alternating regions (good)} & \scalebox{1}{alternating regions (\textcolor{red}{bad})} 

\end{tabular}}

\caption{\label{fig:alt} The formula in \cref{cor:main} applies to non-alternating regions shown on the left. In \cref{sec:alternating-regions}, we give a formula that applies to arbitrary regions, including the alternating regions shown on the right. We indicate the vectors $\v_a,\v_b,\v_c,\v_d$ satisfying~\eqref{eq:alternating} for each alternating region.}
\end{figure}
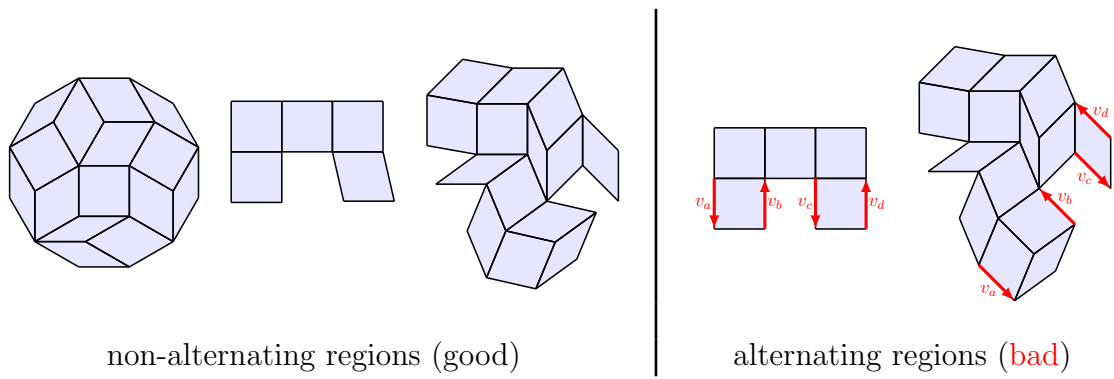

\begin{definition}\label{dfn:curve}
To any region $\Reg$ we associate a curve $\curve_\Reg:\R\to\R^{2n}$ with coordinates $\curve_\Reg(t)=(\g_1(t),\g_2(t),\dots,\g_{2n}(t))$ given by
\begin{equation}\label{eq:curve}
  \g_k(t)= (-1)^{|\J_k\cap [k]|}\prod_{j\in \J_k} \sin(t-\th_j) \qquad\text{for $k\in[2n]$.}
\end{equation}
\end{definition}
\noindent We remark that $\J_k\cap [k]=\{j\in[2n]\mid j<\tau(j)<k\}$. 
 With a more natural ``cyclically symmetric'' choice of conventions~\eqref{eq:aff:curve}, the extra sign $(-1)^{|\J_k\cap [k]|}$ disappears.

\begin{proposition}\label{prop:span_dim_n}
Suppose that $\Reg$ is non-alternating. Then the linear span $\Span(\curve_\Reg)\subset\R^{2n}$ of the vectors $\{\curve_\Reg(t)\}_{t\in\R}$ has dimension $n$.
\end{proposition}
\noindent When $R$ is alternating, $\Span(\curve_\Reg)$ has dimension strictly less than $n$.

Recall that any rhombus tiling of $\Reg$ gives rise to the critical $Z$-invariant Ising model whose boundary correlations are denoted by $\<\sigma_j\sigma_k\>_\Reg$ for $j,k\in[n]$. We assume that the boundary vertex $b_j$ is adjacent to $\v_{2j-1}$ and $\v_{2j}$ for each $j\in[n]$. We introduce the \emph{boundary correlation matrix}
\begin{equation}\label{eq:M_dfn}
  M_\Reg=(m_{j,k}),\quad \text{where} \quad m_{j,k}:=\<\sigma_j\sigma_k\>_\Reg \quad\text{for $j,k\in[n]$}.
\end{equation}
It is a symmetric $n\times n$ matrix.

\begin{figure}
\scalebox{0.86}{
  $   \displaystyle M=\begin{pmatrix}
    1 & m_{12} & m_{13} & m_{14} \\
m_{12} & 1 & m_{23} & m_{24} \\
m_{13} & m_{23} & 1 & m_{34} \\
m_{14} & m_{24} & m_{34} & 1
  \end{pmatrix}\quad \mapsto\quad \double M=
\begin{pmatrix}
1 & 1 & m_{12} & -m_{12} & -m_{13} & m_{13} & m_{14} & -m_{14} \\
-m_{12} & m_{12} & 1 & 1 & m_{23} & -m_{23} & -m_{24} & m_{24} \\
m_{13} & -m_{13} & -m_{23} & m_{23} & 1 & 1 & m_{34} & -m_{34} \\
-m_{14} & m_{14} & m_{24} & -m_{24} & -m_{34} & m_{34} & 1 & 1
\end{pmatrix} $}
  \caption{\label{fig:double} An example of applying the doubling map of~\cite{GP} for $n=4$. Modulo changing some signs in an alternating fashion, $\double M$ contains two copies of each column of $M$. By definition, $\doublemap(M)$ is the row span of $\double M$. See \cref{sec:double}.}
\end{figure}

In our joint work with Pylyavskyy~\cite{GP}, we introduced the \emph{doubling map} $\doublemap$. (See \cref{fig:double} for an example and \cref{sec:double} for the definition.) To any symmetric $n\times n$ matrix $M$ with ones on the diagonal, it associates an $n$-dimensional linear subspace $\doublemap(M)$ of $\R^{2n}$. The map $\doublemap$ is injective: $M$ can be recovered from $\doublemap(M)$ via a simple linear-algebraic procedure outlined below in \cref{cor:main}.  The map $\doublemap$ has remarkable properties: for instance, the Kramers--Wannier duality~\cite{KrWa}, which has a complicated effect on the spin correlations, translates under the map $\doublemap$ into the cyclic shift operator $\R^{2n}\to\R^{2n}$ given in~\eqref{eq:CS}. Similarly, the main result of the present paper can be stated most cleanly in terms of the map $\doublemap$.
\begin{theorem}\label{thm:main_span}
For any non-alternating region $\Reg$, we have
\begin{equation}\label{eq:main_span}
\doublemap(M_\Reg)=\Span(\curve_\Reg).
\end{equation}
\end{theorem}
\noindent An analogous result holds for arbitrary regions; see \cref{thm:arb_reg}.

In order to recover the boundary correlation matrix $M_\Reg$ from $\doublemap(M_\Reg)=\Span(\curve_\Reg)$, one needs to pick a basis of $\Span(\curve_\Reg)$. The easiest way to achieve that is to just take $n$ distinct points on the curve $\curve_\Reg$.
\begin{proposition}\label{prop:rows_dist_pts}
Suppose that $\Reg$ is non-alternating. Then for any $0\leq t_1<t_2<\dots<t_n<\pi$, the vectors $\curve_\Reg(t_1),\curve_\Reg(t_2),\dots,\curve_\Reg(t_n)$ form a basis of $\Span(\curve_\Reg)$.
\end{proposition}
\noindent In the generic case, the proof of \cref{prop:Span_dim=n} below gives a particularly nice set of points. 
\begin{remark}\label{rem:fourier}
A more canonical (and computationally robust) way to produce a basis of $\Span(\curve_\Reg)$ is to observe that each coordinate $\g_j(t)$ is a trigonometric polynomial of degree $n-1$. Therefore it has precisely $n$ non-trivial Fourier coefficients. The rows of the resulting $n\times 2n$ matrix $F_\Reg$ of Fourier coefficients form a basis of $\Span(\curve_\Reg)$ which does not depend on anything besides the angles $\th_1,\th_2,\dots,\th_{2n}$; see \cref{sec:fourier-transform}.
\end{remark}

Lastly, we introduce a $2n\times n$ matrix $K_n$ defined as
\begin{equation}\label{eq:K_n}
K_n=\frac12 \begin{pmatrix}
1 & 0 & \dots & 0\\
1 & 0 & \dots & 0\\
0 & 1 & \dots & 0\\
0 & 1 & \dots & 0\\
\vdots & \vdots & \ddots & \vdots\\
0 & 0 & \dots & 1\\
0 & 0 & \dots & 1
\end{pmatrix}.
\end{equation}
Observe that for $n=4$ and for the $n\times 2n$ matrix $\double M$ from \cref{fig:double}, the product $\double M K_n$ is the $n\times n$ identity matrix. We are ready to translate \cref{thm:main_span} into an explicit matrix formula for boundary correlations of the critical $Z$-invariant Ising model.
\begin{corollary}\label{cor:main}
Let $\Reg$ be a non-alternating region. Choose an $n\times 2n$ matrix $A$ whose row span equals $\Span(\curve_\Reg)$. (See e.g. \cref{prop:rows_dist_pts}, \cref{rem:fourier}, and \cref{thm:arb_reg}.) Let $B=(b_{j,k})$ be the $n\times 2n$ matrix given by
\begin{equation}\label{eq:AKA}
B:=(AK_n)^{-1}A.
\end{equation}
Then, up to a sign, the entries of $B$ are the boundary correlations: we have
\begin{equation}\label{eq:corr_formula_B}
  \<\sigma_j\sigma_k\>_\Reg=\<\sigma_k\sigma_j\>_\Reg=|b_{j,2k-1}|=(-1)^{k-j+1}b_{j,2k-1} \quad\text{for all $1\leq j< k\leq n$.}
\end{equation}
\end{corollary}
\begin{remark}
Part of the content of \cref{cor:main} is that the matrix $AK_n$ is always invertible and that in the notation of \cref{fig:double}, the matrix $B=(AK_n)^{-1}A$ coincides with the matrix $\double M_\Reg$; see \cref{prop:doublemap}. 
\end{remark}

\begin{example}\label{ex:square1}
Consider the case $n=2$ and let $\Reg:=\Reg_n$ be a square. We have 
\begin{equation}\label{eq:}
  \th_1=0,\ \ \th_2=\pi/4,\ \ \th_3=\pi/2,\ \ \th_4=3\pi/4 \quad\text{and}\quad \tau(1)=3,\ \ \tau(2)=4,\ \ \tau(3)=1,\ \ \tau(4)=2.
\end{equation}
Thus $  \J_1=\{2\}$, $\J_2=\{3\}$, $\J_3=\{4\}$, $\J_4=\{1\}$ and 
\begin{equation}\label{eq:}
\curve_{\Reg}(t)= \left(\sin(t-\pi/4),\sin(t-\pi/2),\sin(t-3\pi/4),-\sin(t) \right).
\end{equation}
Let us take $A$ to be the matrix with rows, say, $\curve_{\Reg}(0)$ and $\curve_{\Reg}(3\pi/4)$:
\begin{equation}\label{eq:}
A=\begin{pmatrix}
-\sqrt2/2 & -1 & -\sqrt2/2 & 0\\ 
1 & \sqrt2/2 & 0 & -\sqrt2/2
\end{pmatrix},\quad\text{thus}\quad AK_n=\frac12\begin{pmatrix}
-1-\sqrt2/2 & -\sqrt2/2\\
1+\sqrt2/2 & -\sqrt2/2
\end{pmatrix}.
\end{equation}
We calculate
\begin{equation}\label{eq:}
(AK_n)^{-1}=\begin{pmatrix}
\sqrt2-2 & 2-\sqrt2\\
-\sqrt2 & -\sqrt2
\end{pmatrix} \quad\text{and}\quad (AK_n)^{-1}A=\begin{pmatrix}
1 & 1 & \sqrt2-1 & 1-\sqrt2\\
 1-\sqrt2 & \sqrt2-1 & 1 & 1
\end{pmatrix}.
\end{equation}
By \cref{cor:main}, we find $\<\sigma_1\sigma_2\>=\sqrt2-1$. This is indeed the correct value; see \cref{ex:square2}. It is also consistent with \cref{thm:reg}.
\end{example}

\begin{remark}
 The formula~\eqref{eq:AKA} involves inverting the matrix $AK_n$. In general, correlations of the Ising model may be computed in terms of the Kasteleyn matrix or the Kac--Ward matrix associated with $G$~\cite{KacWard,Kasteleyn,Fisher2,Fisher,DZMSS}. These matrices depend on the choice of the rhombus tiling and are roughly of size $n^2\times n^2$. On the other hand, the matrix $AK_n$ is of size $n\times n$ and depends manifestly only on the angles $\th_1,\th_2,\dots,\th_{2n}$ (cf. \cref{rem:fourier}). 
\end{remark}

\begin{remark}
Many of our constructions are special cases of objects arising in the total positivity literature. In the main body of the paper, we present them in a self-contained way, and explain their relations to total positivity in \cref{sec:relat-total-posit}.
\end{remark}

\begin{remark}
\Cref{thm:main_span} and most of our other results generalize directly to the setting of the totally nonnegative Grassmannian $\Grtnn(k,n)$ studied in~\cite{Lus98,Pos}: one replaces a matching with a decorated permutation, a rhombus tiling with a plabic tiling~\cite{Pos,OPS}, the critical $Z$-invariant Ising model with a suitable critical dimer model~\cite{KenCrit} on a plabic graph, and the doubling map $\doublemap$ with Postnikov's boundary measurement map. This simultaneously includes the cases of the Ising model (for $\Grtnn(n,2n)$) and electrical resistor networks (for $\Grtnn(n-1,2n)$; cf.~\cite{Lam}), providing partial progress towards~\cite[Question~9.2]{GP}. These results will appear in a forthcoming paper~\cite{GalPrep}.
\end{remark}

\begin{remark}\label{rmk:KW}
We showed in~\cite[Proposition~3.6]{GP} that for each $n\geq 2$, there exists a \emph{unique} $n\times n$ boundary correlation matrix $M_0$ that is invariant under Kramers--Wannier duality~\cite{KrWa}. Applying this duality to the graph $G_\Tiling$ associated with a rhombus tiling $\Tiling$ amounts to switching the black/white colors of vertices of $\Tiling$ and relabeling $(\v_1,\v_2,\dots,\v_{2n})\mapsto(\v_2,\dots,\v_{2n},\v_1)$. For the case when $\Reg=\Reg_n$ is a regular $2n$-gon, any rhombus tiling can be connected by a sequence of flips to its rotation by $\pi/n$. Thus $M_{\Reg_n}$ coincides with the unique self-dual $n\times n$ matrix $M_0$.
\end{remark}

\subsection*{Acknowledgments} I am indebted to Pasha Pylyavskyy for his numerous contributions at various stages of this project. I also thank Cl\'ement Hongler for bringing several useful references to my attention. In addition, I am grateful to the anonymous referee for their valuable suggestions. This work was supported by an Alfred P. Sloan Research Fellowship and by the National Science Foundation under Grants No.~DMS-1954121 and No.~DMS-2046915.

\section{Background}

\subsection{Matchings, regions, tilings, and pseudoline arrangements}\label{sec:match-regi-tilings}
Since our proof will necessarily pass through very degenerate regions and their rhombus tilings, we need to define these objects formally.

For $1\leq j<k\leq 2n$, we introduce \emph{cyclic intervals} $[j,k]:=\{j,j+1,\dots,k\}$ and $[j,k]^c:=[2n]\setminus[j,k]$. By a \emph{matching} we mean a map $\tau:[2n]\to[2n]$ such that $\tau(j)=k$ implies $\tau(k)=j$ and $j\neq k$. (Such a map is also called a \emph{fixed-point-free involution}.) We say that a sequence $\bth=(\th_1,\th_2,\dots,\th_{2n})$ is a \emph{$\tau$-shape} if it satisfies~\eqref{eq:th+pi/2} and~\eqref{eq:th<th<th<th} for all $j,k\in[2n]$.  For a $\tau$-shape $\bth$, we let the vectors $\v_1,\v_2,\dots,\v_{2n}\in\C$ be defined by~\eqref{eq:v=exp(th)}. By a \emph{valid region} we mean a pair $\Reg=(\tau,\bth)$ consisting of a matching $\tau$ and a $\tau$-shape $\bth$. We draw $\Reg$ in the plane as a polygonal chain with sides $\v_1,\v_2,\dots,\v_{2n}$ (in this order). By~\eqref{eq:v=-v}, this polygonal chain is closed. In general, it may intersect itself.

Let us say that a matching $\tau:[2n]\to[2n]$ is \emph{disconnected} if there exist $1\leq j<k\leq 2n$ such that $\tau([j,k])=[j,k]$ but $[j,k]\neq [2n]$; otherwise $\tau$ is called \emph{connected}. When $\tau([j,k])=[j,k]$, the \emph{restriction} of $\tau$ to $[j,k]$ and $[j,k]^c$ is defined in an obvious way. By~\eqref{eq:v=-v}, $\tau([j,k])=[j,k]$ implies $\v_j+\v_{j+1}+\dots+\v_k=0$, so if $\tau$ is disconnected then the boundary of $\Reg$ is self-intersecting. A \emph{connected component} of $\tau$ is a cyclic interval $I$ such that the restriction of $\tau$ to $I$ is connected. The conditions~\eqref{eq:th+pi/2} and~\eqref{eq:th<th<th<th} only involve pairs of indices from the same connected component. All our constructions will work independently for each connected component of $\tau$.

We say that a valid region $\Reg$ is \emph{simple} if its boundary is a  simple (non-self-intersecting) polygonal chain. In this case, $\tau$ must be connected, but the converse need not hold. A \emph{rhombus tiling} of a simple region $\Reg$ is a finite collection of rhombi whose sides are unit vectors parallel to the vectors $\v_1,\v_2,\dots,\v_{2n}$, whose interiors do not overlap, and whose union is the region bounded by $\Reg$. The \emph{critical $Z$-invariant Ising model} defined below is most naturally associated to a rhombus tiling of a given simple region, but in fact it is well defined for any valid region in terms of \emph{pseudoline arrangements}. 

Consider $2n$ points $d_1,d_2,\dots,d_{2n}$ drawn counterclockwise on a circle. A \emph{pseudoline arrangement} is a collection $\Arr$ of $n$ embedded line segments (``pseudolines'') such that any two pseudolines intersect at most once, all intersection points are in the interior of the disk, and no three pseudolines intersect at one point. The set of endpoints of the pseudolines is assumed to be $\{d_1,d_2,\dots,d_{2n}\}$ (thus the endpoints are pairwise distinct). Each pseudoline connects some $d_j$ to $d_k$ and this gives rise to a matching $\tau=\tau_\Arr:[2n]\to[2n]$ sending $j\mapsto k$ and $k\mapsto j$ for each such pseudoline. The procedure in \cref{sec:arbitrary-regions} describes how pseudoline arrangements are naturally the dual objects to rhombus tilings; see \figref{fig:intro2}(c). It is easy to see that any matching is realized by some pseudoline arrangement: one can choose generic boundary points $d_1,d_2,\dots,d_{2n}$ and take each pseudoline to be just a straight line segment connecting $d_j$ to $d_{\tau(j)}$.

\subsection{Ising model}\label{sec:ising-model}
Let $G=(V,E)$ be a finite simple undirected graph. A \emph{spin configuration} is an assignment $\sigma=(\sigma_v)_{v\in V}\in\{\pm1\}^V$ of spins to the vertices of $G$, where we have $\sigma_v=\pm1$ for each $v\in V$. Given an assignment $\bx=\{x_e\}_{e\in E}$ of positive real edge weights, the Ising model is a probability distribution on the set $\{\pm1\}^V$ of all spin configurations: the probability of a given spin configuration $\sigma$ equals
\begin{equation}\label{eq:Z_dfn}
  \Prob(\sigma):=\frac1Z\prod_{\{u,v\}\in E:\, \sigma_u=\sigma_v} x_{\{u,v\}},\quad\text{where}\quad Z:=\sum_{\sigma\in\{\pm1\}^V} \prod_{\{u,v\}\in E:\, \sigma_u=\sigma_v} x_{\{u,v\}}
\end{equation}
is the \emph{partition function}. Given two vertices $u,v\in V$, we define their \emph{correlation} as
\begin{equation}\label{eq:corr_dfn}
\<\sigma_u\sigma_v\>:=\Prob(\sigma_u=\sigma_v)-\Prob(\sigma_u\neq\sigma_v).
\end{equation}

Suppose we are given a rhombus tiling $\Tiling$ of a simple region $\Reg=(\tau,\bth)$. Color the vertices of $\Tiling$ in a bipartite way so that the vertex adjacent to $\v_{2j-1}$ and $\v_{2j}$ is black for each $j\in[n]$. The graph $G_\Tiling=(V,E)$ is defined as follows: the vertex set $V$ consists of all black vertices of $\Tiling$, and the edge set $E$ contains, for each rhombus in $\Tiling$, the diagonal connecting its two black vertices. For an edge $e\in E$, let $2\th_e\in(0,\pi)$ be the angle at a white vertex of the rhombus containing $e$. The edge weights $\bx_\Tiling=(x_e)_{e\in E}$ are defined as follows: for $e\in E$, we set $x_e:=\cot(\th_e/2)\in(1,\infty)$; see \figref{fig:x_e}(a). The Ising model associated with the weighted graph $(G_\Tiling,\bx_\Tiling)$ is referred to as the \emph{critical $Z$-invariant Ising model}.

\begin{example}\label{ex:square2}
Consider the case $n=2$ and $\Reg=\Reg_n$. Then $G$ is a single edge $e$ connecting $b_1$ to $b_2$. The rhombus containing $e$ has all angles equal to $\pi/2$, thus $\th_e=\pi/4$, and $x_e=\cot(\pi/8)=\sqrt2+1$. By~\eqref{eq:Z_dfn} and~\eqref{eq:corr_dfn}, we have
\begin{equation}\label{eq:}
  Z=x_e+1=\sqrt2+2 \quad\text{and}\quad\<\sigma_1\sigma_2\>_{\Reg}=\frac{x_e-1}{x_e+1}=\frac{\sqrt2}{\sqrt2+2}=\sqrt2-1,
\end{equation}
in agreement with \cref{ex:square1} and \cref{thm:reg}.
\end{example}

A \emph{flip} is a local operation $\Tiling\mapsto\Tiling'$ on rhombus tilings which replaces three rhombi in $\Tiling$ whose union is a convex hexagon with the other three rhombi whose union is the same hexagon. The associated weighted graphs $(G_\Tiling,\bx_\Tiling)$ and $(G_{\Tiling'},\bx_{\Tiling'})$ are related by a \emph{star-triangle move}; see \figref{fig:intro1}(d). The edge weights $\bx_\Tiling,\bx_{\Tiling'}$ have the important property that applying a flip preserves the correlations $\<\sigma_u\sigma_v\>$ whenever $u,v$ are black vertices present in both $\Tiling$ and $\Tiling'$. (We caution that the partition function $Z$ is in general \emph{not} preserved by star-triangle moves.)

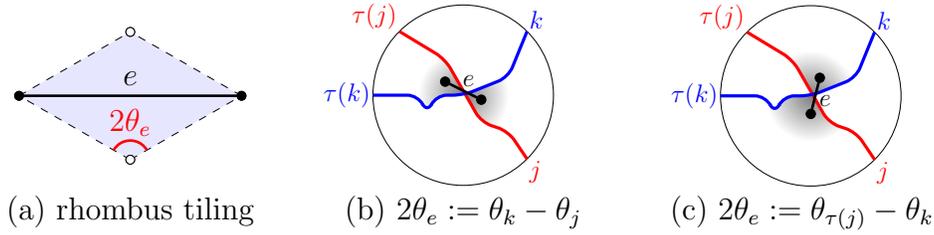
\begin{figure}

\def\sq{0.70710}
\def\rc{10pt}
\def\lw{0.5pt}
\def\lww{3pt}
\def\scl{0.4}
\def\nodescl{2}
\def\elw{2.5pt}
\def\bwscl{0.8}
\def\escl{2}
\def\scll{1}
\def\isp{1pt}
\def\botscl{2.5}
\def\xescl{1}

\def\op{0.1}

\def\centerarc[#1](#2)(#3:#4:#5)
    { \draw[#1] ($(#2)+({#5*cos(#3)},{#5*sin(#3)})$) arc (#3:#4:#5); }

\setlength{\tabcolsep}{10pt}

\begin{tabular}{ccc}

\scalebox{\scll}{
\begin{tikzpicture}[scale=1.7,baseline=(ZUZU.base)]
\coordinate (0) at (0.00,0.00);
\coordinate (2) at (0.87,0.50);
\coordinate (6) at (-0.87,0.50);
\coordinate (26) at (0.00,1.00);

\node[fill=white,circle,draw=black,line width=0.3pt,scale=0.3] (A0) at (0.00,0.00) {};
\node[fill=black,circle,draw=black,line width=0.3pt,scale=0.3] (A2) at (0.87,0.50) {};
\node[fill=black,circle,draw=black,line width=0.3pt,scale=0.3] (A6) at (-0.87,0.50) {};
\node[fill=white,circle,draw=black,line width=0.3pt,scale=0.3] (A26) at (0.00,1.00) {};
\draw[dashed,line width=0.1pt] (0) -- (2);
\draw[dashed,line width=0.1pt] (0) -- (6);
\fill [opacity=\op,blue] (0.center)-- (2.center)-- (26.center)-- (6.center) -- cycle;
\draw[black, line width=1pt] (A6) -- (A2);
\draw[dashed,line width=0.1pt] (2) -- (26);
\draw[dashed,line width=0.1pt] (6) -- (26);
\node[fill=white,circle,draw=black,line width=0.3pt,scale=0.3] (A0) at (0.00,0.00) {};
\node[fill=black,circle,draw=black,line width=0.3pt,scale=0.3] (A2) at (0.87,0.50) {};
\node[fill=black,circle,draw=black,line width=0.3pt,scale=0.3] (A6) at (-0.87,0.50) {};
\node[fill=white,circle,draw=black,line width=0.3pt,scale=0.3] (A26) at (0.00,1.00) {};

\centerarc[red,line width=1pt](0,0)(30:150:0.15)
\node[scale=1,red,anchor=south] (A) at (0,0.1) {$2\theta_e$};
\node[scale=1,black,anchor=south] (A) at (0,0.5) {$e$};

\node[scale=1,black,anchor=north] (ZUZU) at (0,-0.2) {(a) rhombus tiling};

\end{tikzpicture}}

&

\scalebox{\scl}{
\begin{tikzpicture}[scale=3,baseline=(ZUZU.base)]
\begin{scope}[even odd rule]
\clip[rounded corners=\rc] (1.5,1.5)--(\sq,\sq)--(0.5,0.2)--(0,0)--(-0.3,0)--(-0.4,-0.2)--(-0.5,0)--(-1.5,0)--(-1.5,1.5)--(1.5,1.5)--(-1.5,1.5)--(-\sq,\sq)--(-0.2,0.4)--(0.2,-0.3)--(0.5,-0.4)--(\sq,-\sq)--(1.5,-1.5)--cycle;
\shade[inner color=black!50,outer color=white] (0,0) circle (0.5);
\end{scope}
\draw[line width=\lw] (0,0) circle (1);
\draw[rounded corners=\rc, line width=\lww,red] (\sq,-\sq)--(0.5,-0.4)--(0.2,-0.3)--(-0.2,0.4)--(-\sq,\sq);
\draw[rounded corners=\rc, line width=\lww,blue] (\sq,\sq)--(0.5,0.2)--(0,0)--(-0.3,0)--(-0.4,-0.2)--(-0.5,0)--(-1,0);
\node[scale=\nodescl,anchor=north west,red,inner sep=\isp] (J) at (\sq,-\sq) {$j$};
\node[scale=\nodescl,anchor=south east,red,inner sep=\isp] (TJ) at (-\sq,\sq) {$\tau(j)$};
\node[scale=\nodescl,anchor=east,blue,inner sep=\isp] (TK) at (-1,0) {$\tau(k)$};
\node[scale=\nodescl,anchor=south west,blue,inner sep=\isp] (K) at (\sq,\sq) {$k$};
\node[scale=\bwscl,draw=black,circle,fill=black] (A) at (-0.2,0.15) {};
\node[scale=\bwscl,draw=black,circle,fill=black] (B) at (0.2,-0.05) {};
\draw[line width=\elw,black] (A)--(B);
\node[scale=\escl,black,anchor=south west] (E) at (-0.1,0){$e$};
\node[scale=\botscl](ZUZU) at (0,-1.3) {(b) $2\th_e:=\th_k-\th_j$};
\end{tikzpicture}
}

&

\scalebox{\scl}{
\begin{tikzpicture}[scale=3,baseline=(ZUZU.base)]
\begin{scope}[even odd rule]
\clip[rounded corners=\rc] (1.5,1.5)--(\sq,\sq)--(0.5,0.2)--(0,0)--(-0.3,0)--(-0.4,-0.2)--(-0.5,0)--(-1.5,0)--(-1.5,-1.5)--(1.5,-1.5)--(\sq,-\sq)--(0.5,-0.4)--(0.2,-0.3)--(-0.2,0.4)--(-\sq,\sq)--(-1.5,1.5)--(1.5,1.5)--cycle;
\shade[inner color=black!50,outer color=white] (0,0) circle (0.5);
\end{scope}
\draw[line width=\lw] (0,0) circle (1);
\draw[rounded corners=\rc, line width=\lww,red] (\sq,-\sq)--(0.5,-0.4)--(0.2,-0.3)--(-0.2,0.4)--(-\sq,\sq);
\draw[rounded corners=\rc, line width=\lww,blue] (\sq,\sq)--(0.5,0.2)--(0,0)--(-0.3,0)--(-0.4,-0.2)--(-0.5,0)--(-1,0);
\node[scale=\nodescl,anchor=north west,red,inner sep=\isp] (J) at (\sq,-\sq) {$j$};
\node[scale=\nodescl,anchor=south east,red,inner sep=\isp] (TJ) at (-\sq,\sq) {$\tau(j)$};
\node[scale=\nodescl,anchor=east,blue,inner sep=\isp] (TK) at (-1,0) {$\tau(k)$};
\node[scale=\nodescl,anchor=south west,blue,inner sep=\isp] (K) at (\sq,\sq) {$k$};
\node[scale=\bwscl,draw=black,circle,fill=black] (A) at (-0,-0.2) {};
\node[scale=\bwscl,draw=black,circle,fill=black] (B) at (0.1,0.2) {};
\draw[line width=\elw,black] (A)--(B);
\node[scale=\escl,black,anchor=north west] (E) at (-0,0.1){$e$};
\node[scale=\botscl](ZUZU) at (-0.1,-1.3) {(c) $2\th_e:=\th_{\tau(j)}-\th_k$};
\end{tikzpicture}
}

\end{tabular}

  \caption{\label{fig:x_e} In each case, we assign $x_e:=\cot(\th_e/2)$. (a) For a rhombus in a rhombus tiling, $2\th_e$ is the angle at a white vertex; (b) and (c) refer to the case of an edge $e$ in the graph $G_\Arr$ associated to a pseudoline arrangement $\Arr$; here $j<k<\tau(j)<\tau(k)$.}
\end{figure}

More generally, given any valid region $\Reg=(\tau,\bth)$, choose a pseudoline arrangement $\Arr$ such that $\tau=\tau_\Arr$. (For example, $\Arr$ may be chosen to consist of straight line segments as above.) The complement of $\Arr$ in the disk may be colored black and white in a checkerboard fashion so that the arc connecting $d_{2j-1}$ to $d_{2j}$ is adjacent to a black region for each $j\in[n]$; see e.g.~\cite[Figure~1]{Bax}. We now let $G_\Arr=(V,E)$ be the graph whose vertices are the black regions of $\Arr$ (including $n$ boundary regions $b_1,b_2,\dots,b_n$, where $b_j$ is adjacent to the arc between $d_{2j-1}$ and $d_{2j}$ for $j\in[n]$), and whose edges correspond to intersection points between the pseudolines in $\Arr$. Each such intersection point $p$ involves a pseudoline connecting $j$ to $\tau(j)$ and a pseudoline connecting $k$ to $\tau(k)$ for some $j<k<\tau(j)<\tau(k)$. The corresponding edge $e$ connects (the vertices of $G$ corresponding to) the two black regions adjacent to $p$. We set $x_e:=\cot(\th_e/2)$, where we either have $2\th_e:=\th_k-\th_j$ or $2\th_e:=\th_{\tau(j)}-\th_k$, depending on how the two black regions are located relative to $j,k,\tau(j),\tau(k)$; see \figref{fig:x_e}(b,c). We set $\bx_\Arr:=(x_e)_{e\in E}$. For each triangular interior region of $\Arr$, one may perform a \emph{Yang--Baxter} move which is dual to a flip of a rhombus tiling. The associated weights still satisfy a similar star-triangle relation which preserves the boundary correlations of the associated critical $Z$-invariant Ising model. If a single black region contains several boundary points $b_{j_1},\dots,b_{j_k}$, they are treated as if they were ``contracted'' into a single vertex, and we set $\<\sigma_{j_s}\sigma_{j_t}\>:=1$ for all $s,t\in[k]$. See~\cite[Definition~6.1]{GP} for details.

The above construction associates a weighted graph $(G_\Tiling,\bx_\Tiling)$ (resp., $(G_\Arr,\bx_\Arr)$) to any rhombus tiling $\Tiling$ (resp., pseudoline arrangement $\Arr$) of a valid region $\Reg=(\tau,\bth)$. The graph $G_\Tiling$ has $n$ boundary vertices $b_1,b_2,\dots,b_n$, and the Ising model associated with $(G_\Tiling,\bx_\Tiling)$ yields an $n\times n$ boundary correlation matrix $M_\Reg$ defined in~\eqref{eq:M_dfn}. The construction of $(G_\Tiling,\bx_\Tiling)$ depends on a rhombus tiling $\Tiling$ (or a pseudoline arrangement $\Arr$), but the resulting boundary correlation matrix $M_\Reg$ depends only on $\Reg$.

\subsection{The doubling map}\label{sec:double}
We describe the map $\doublemap$ introduced in~\cite{GP}. For a symmetric $n\times n$ matrix $M=(m_{j,k})$ with ones on the diagonal, introduce an $n\times 2n$ matrix $\double M=(\double m_{j,k})$ satisfying $|m_{j,k}|=|\double m_{j,2k-1}|=|\double m_{j,2k}|$ for all $j,k\in[n]$; see \cref{fig:double}. The signs are chosen in a simple alternating fashion: for $j=k$, set $\double m_{j,2k-1}=\double m_{j,2k}:=m_{j,j}=1$, and for $j\neq k$, set
\begin{equation}\label{eq:double_signs}
  \double m_{j,2k-1}=-\double m_{j,2k}:=(-1)^{j+k+\one(j<k)}m_{j,k},\quad\text{where}\quad \one(j<k):=
  \begin{cases}
    1, &\text{if $j<k$,}\\
    0,&\text{otherwise.}
  \end{cases}
\end{equation}
We let $\doublemap(M)\subset\R^{2n}$ denote the linear subspace spanned by the rows of $\double M$.  Recall that we have introduced a $2n\times n$ matrix $K_n$ in~\eqref{eq:K_n}.

\begin{proposition}[\cite{GP}]\label{prop:doublemap}
Let $M$ be a symmetric $n\times n$ matrix with ones on the diagonal.
\begin{theoremlist}
\item\label{item:doublemap:dim} The subspace $\doublemap(M)$ has dimension $n$.
\item\label{item:doublemap:AKA} For any $n\times 2n$ matrix $A$ whose rows form a basis of $\doublemap(M)$, the matrix $AK_n$ is invertible and satisfies
\begin{equation}\label{eq:AKA_GP}
  (AK_n)^{-1}A=\double M.
\end{equation}
\end{theoremlist}
\end{proposition}
\begin{proof}
By construction, $\double M K_n$ is the $n\times n$ identity matrix. Both of the above statements now follow in a straightforward way. Alternatively, see the proof of~\cite[Lemma~6.7]{GP}.
\end{proof}
\noindent  Thus \cref{cor:main} follows from \cref{thm:main_span} via \cref{prop:doublemap}.

The image of the space of $n\times n$ Ising boundary correlation matrices under the map $\doublemap$ is invariant under the \emph{cyclic shift} operator
\begin{equation}\label{eq:CS}
S:\R^{2n}\to\R^{2n},\quad   (x_1,x_2,\dots,x_{2n})\mapsto \left((-1)^{n-1}x_{2n},x_1,x_2,\dots,x_{2n-1}\right).
\end{equation}
The sign twist is chosen in such a way that $S$ preserves the totally nonnegative Grassmannian; see~\cite[Remark~3.3]{Pos}. By~\cite[Theorem~3.4]{GP}, $S$ is the image of the Kramers--Wannier duality (cf. \cref{rmk:KW}) under the map $\doublemap$.

\subsection{Affine notation}\label{sec:affine-notation}
All of our constructions respect the cyclic shift action~\eqref{eq:CS}. At times, it will be more convenient to use notation that is invariant under this cyclic symmetry.

Given a matching $\tau:[2n]\to[2n]$, we extend it to the unique bijection $\taut:\Z\to\Z$ satisfying the following conditions:
\begin{enumerate}
\item $\taut(k+2n)=\taut(k)+2n$ for all $k\in\Z$;
\item $k<\taut(k)<k+2n$ for all $k\in\Z$;
\item $\taut(k)\equiv\tau(k)\pmod{2n}$ for all $k\in[2n]$.
\end{enumerate}
For example, we see that $\taut^2(k)=k+2n$ for all $k\in\Z$. Similarly, we extend $\bth$ to the unique sequence $\btht=(\tht_k)_{k\in\Z}$ satisfying $\tht_k=\th_k$ for $k\in[2n]$ and $\tht_{k+2n}=\tht_k+\pi$ for all $k\in\Z$. For $k\in\Z$, we set $\vt_k:=\exp(2i\tht_k)$, which satisfies $\vt_{k+2n}=\vt_k$ for all $k\in\Z$. We also have ``affine analogs'' of~\eqref{eq:th+pi/2}--\eqref{eq:th<th<th<th}:
\begin{align}
\retainlabel{eq:aff:th+pi/2} \label{eq:aff:th+pi/2}
\tht_{\taut(k)}&=\tht_k+\pi/2, &&\text{for all $k\in\Z$;}\\
\retainlabel{eq:aff:th<th<th<th}\label{eq:aff:th<th<th<th}
\tht_j<\tht_k&<\tht_{\taut(j)}<\tht_{\taut(k)}, &&\text{for all $j,k\in\Z$ satisfying $j<k<\taut(j)<\taut(k)$.}
\end{align}
Finally, for each $k\in\Z$, we set
\begin{equation}\label{eq:aff:Jt}
\Jt_k=\{\taut(j)\mid j\in\Z\text{ such that $j<k$ and $\taut(j)>k$}\}.
\end{equation}
We find that for $k\in[2n]$, the sets $\Jt_k$ and $\J_k$ coincide modulo $2n$. We can now explain the signs appearing in~\eqref{eq:curve}: for $k\in[2n]$, we have 
\begin{equation}\label{eq:aff:curve}
\g_k(t)=\prod_{j\in\Jt_k} \sin(t-\tht_j).
\end{equation}
We extend the definition of $\g_k(t)$ to all $k\in\Z$ using~\eqref{eq:aff:curve}. It satisfies $\g_{k+2n}(t)=(-1)^{n-1}\g_k(t)$, in agreement with~\eqref{eq:CS}.

Clearly,  $\tau$ (resp., $\bth$) determines and is uniquely determined by its affine analog $\taut$ (resp.,~$\btht$).  In what follows, we switch freely between the two conventions.

\subsection{Removing a crossing}\label{sec:removing-crossing}
Given a matching $\tau$, we define 
\begin{equation}\label{eq:}
\xing(\tau):=\#\{j,k\in[2n]\mid j<k<\tau(j)<\tau(k)\}.
\end{equation}
 If $\xing(\tau)=0$ then $\tau$ is called \emph{non-crossing}. In this case, the graph $G$ has no edges, and the correlation matrix $M_\Reg$ consists of zeroes and ones. Otherwise, we can find an index $k\in[2n]$ satisfying $k<k+1<\taut(k)<\taut(k+1)$. We call such an index $k$ \emph{a $\tau$-descent}.

\begin{definition}
Let $\Reg=(\tau,\bth)$ be a valid region and $k\in[2n]$ be a $\tau$-descent. We introduce another pair $\ska\Reg=(\tau',\bth')$ defined as follows: we put 
\begin{equation}\label{eq:}
\taut'(k):=\taut(k+1),\quad \taut'(k+1):=\taut(k),\quad \taut'(\taut(k)):=k+1+2n,\quad \taut'(\taut(k+1)):=k+2n,
\end{equation}
and $\taut'(j):=\taut(j)$ for $j\in\Z$ not equal to one of $k$, $k+1$, $\tau(k)$, or $\tau(k+1)$ modulo $2n$. Similarly, $\btht'=(\tht'_j)_{j\in\Z}$ is defined by setting $\tht'_k:=\tht_{k+1}$, $\tht'_{k+1}:=\tht_k$, and $\tht_j:=\th_j$ for all $j\in\Z$ not equal to either $k$ or $k+1$ modulo $2n$.
\end{definition}

The following is straightforward to check.
\begin{proposition}
If $\Reg=(\tau,\bth)$ is a valid region and $k\in[2n]$ is a $\tau$-descent then $\ska\Reg$ is also a valid region. If $\Reg$ is non-alternating then so is $\ska \Reg$.
\qed
\end{proposition}

\def\gth{g^{\bth}}
\def\Bth{B^{\bth}}
In order to make an inductive argument, we need one more ingredient from~\cite{GP}. Let $\Reg=(\tau,\bth)$ be a valid region. For each $\tau$-descent $k\in[2n]$, we will define a $2n\times 2n$ matrix $\gth_k$. Denote $s_k:=\sin(\tht_{k+1}-\tht_k)$ and $c_k:=\cos(\tht_{k+1}-\tht_k)$ and let 
\begin{equation}\label{eq:}
\Bth_k:=\begin{pmatrix}
1/c_k & s_k/c_k\\
s_k/c_k & 1/c_k
\end{pmatrix}.
\end{equation}
For $k<2n$, the matrix $\gth_k$ coincides with the $2n\times 2n$ identity matrix except for a $2\times 2$ block $\Bth_k$ which appears in rows and columns $k,k+1$. For $k=2n$, the matrix $\gth_{2n}$ coincides with the $2n\times 2n$ identity matrix except for the following four entries: 
\begin{equation}\label{eq:}
  (\gth_{2n})_{1,1}=(\gth_{2n})_{2n,2n}=1/c_{2n}, \quad\text{and}\quad (\gth_{2n})_{1,2n}=(\gth_{2n})_{2n,1}=(-1)^{n-1}s_{2n}/c_{2n}.
\end{equation}
Clearly, for a $\tau$-descent $k\in[2n]$, we have $\taut(k)\neq k+1$. Furthermore, by~\eqref{eq:aff:th<th<th<th}, we must have $\tht_k<\tht_{k+1}<\tht_k+\pi/2$. In particular, $c_{k}\neq0$ for any $\tau$-descent $k$.
\begin{example}
For $n=2$, we have
\[\gth_1(t)=\begin{pmatrix}
1/c_1 & s_1/c_1 & 0 & 0\\
s_e/c_1 & 1/c_1 & 0 & 0\\
0 & 0 & 1 & 0\\
0 & 0 & 0 & 1\\
\end{pmatrix},
\quad
\gth_4(t)=\begin{pmatrix}
1/c_4 & 0 & 0 & -s_4/c_4\\
0 & 1 & 0 & 0\\
0 & 0 & 1 & 0\\
-s_4/c_4 & 0 & 0 & 1/c_4
\end{pmatrix}.
\]
\end{example}

The above $2n\times 2n$ matrices represent linear operators on $\R^{2n}$. Given an $n$-dimensional subspace $X\subset \R^{2n}$ and a $2n\times 2n$ matrix $g$, the subspace $X\cdot g$ is defined as $\{x\cdot g\mid x\in X\}$, where $x$ is treated as a row vector.
\begin{proposition}[{\cite[Theorem~3.22]{GP}}]
Let $\Reg=(\tau,\bth)$ be a valid region and $k\in[2n]$ be a $\tau$-descent. Then for $\Reg':=\ska\Reg$, we have 
\begin{equation}\label{eq:doublemap_ind}
\doublemap(M_\Reg)=\doublemap(M_{\Reg'})\cdot \gth_k.
\end{equation}
\end{proposition}

\begin{remark}
Suppose that a rhombus tiling $\Tiling$ contains a rhombus adjacent to the boundary edges $\vt_k$ and $\vt_{k+1}$ for some $k\in[2n]$, and let $e\in E$ be the edge of $G_\Tiling$ inside this rhombus. To this edge, one can associate two quantities $s_e$ and $c_e$; see~\cite[Eq.~(3.6)]{GP}. When $k$ is odd, the vertex of $\Tiling$ adjacent to $v_k$ and $v_{k+1}$ is black, and when $k$ is even, it is white. Therefore we see that the weight $x_e$ of $e$ is equal to $\cot((\tht_{k+1}-\tht_k)/2)$ when $k$ is odd and to $\cot(\pi/4-(\tht_{k+1}-\tht_k)/2)$ when $k$ is even. Comparing this to~\cite[Eq.~(3.6)]{GP}, we find that $s_e=s_k$, $c_e=c_k$ when $k$ is odd and $s_e=c_k$, $c_e=s_k$ when $k$ is even. In~\cite{GP},  $s_e$ and $c_e$ were swapped in the definition of $\gth_k$ whenever $k$ was even. Therefore we do not ever need to swap $s_k$ and $c_k$ in our definition of $\gth_k$.
\end{remark}

\subsection{Real and complex subspaces}\label{sec:real-compl-subsp}
The space $\R^{2n}$ is naturally embedded as a subset of~$\C^{2n}$. If $V\subset\R^{2n}$ is an $n$-dimensional linear subspace, one can consider its $\C$-span $V_\C\subset\C^{2n}$. Conversely, given an $n$-dimensional subspace $U\subset\C^{2n}$ that is invariant under conjugation, its intersection with $\R^{2n}\subset\C^{2n}$ will be an $n$-dimensional real subspace. Given any complex $n\times 2n$ matrix $A$ whose \emph{$\C$-row span} (i.e., row span with complex coefficients) coincides with $V_\C$ for some $V\subset\R^{2n}$, and given any real $2n\times n$ matrix $K$ such that $AK$ is invertible, the matrix $(AK)^{-1}A$ will be real and its $\R$-row  span will be equal to $V$. Therefore in \cref{cor:main}, it suffices to find any complex $n\times 2n$ matrix $A$ whose $\C$-row span coincides with $\SpanC(\curve_\Reg):=(\Span(\curve_\Reg))_\C$. We will present a canonical such matrix in the next section.
In what follows, we will freely switch between $\R$-spans and $\C$-spans.

\section{Fourier transform}\label{sec:fourier-transform}
Let $\Reg$ be a valid region. Recall the expression for $\curve_\Reg(t)=(\g_1(t),\g_2(t),\dots,\g_{2n}(t))$ from~\eqref{eq:aff:curve}. Denote $T:=\exp(it)$ and $T_k:=\exp(i\tht_k)$ for each $k\in\Z$. Using the formula $\sin(x)=\frac{\exp(ix)-\exp(-ix)}{2i}$, we get
\begin{equation}\label{eq:gamma_fourier}
  \g_k(t)=\frac{1}{(2i)^{n-1}}\prod_{m\in \Jt_k}  \left(\frac{T}{T_m}-\frac{T_m}{T}\right)=\frac{1}{(2i)^{n-1}}\sum_{j=1}^{n} (-1)^{n-j} \coef_{j,k} T^{2j-n-1}.
\end{equation}
Here each coefficient $\coef_{j,k}$ is, up to a constant that depends only on $k$, the $(j-1)$-th elementary symmetric polynomial\footnote{By definition, we have $e_{a}(x_1,x_2,\dots,x_b):=\sum_{1\leq j_1<j_2<\dots<j_{a}\leq b} x_{j_1}x_{j_2}\cdots x_{j_{a}}$.} in the variables $(T_m^2)_{m\in \Jt_k}$, and we let $F_\Reg$ be the corresponding~${n\times 2n}$  \emph{Fourier coefficient matrix}:
\begin{equation}\label{eq:fourier_coef}
  F_\Reg:=(\coef_{j,k}), \quad\text{where}\quad  \coef_{j,k}= \frac{e_{j-1}((T_m^2)_{m\in \Jt_k})}{\prod_{m\in \Jt_k} T_m} \quad\text{for $j\in[n],k\in[2n]$}.
\end{equation}
\begin{example}
In the case $R:=\Reg_2$ of \cref{ex:square1}, we have 
\begin{equation}\label{eq:}
T_1=1,\ \ T_2=\frac{1+i}{\sqrt2},\ \ T_3=i,\ \ T_4=\frac{-1+i}{\sqrt2},\ \ T_5=-T_1,\ \ \dots\;.
\end{equation}
We have $e_0(T^2_j)=1$ and $e_1(T_j^2)=T_j^2$, thus~\eqref{eq:gamma_fourier} yields
\begin{equation}\label{eq:}
F_\Reg=\begin{pmatrix}
1/{T_2}&1/{T_3}&1/{T_4}&-1/{T_1}\\
T_2 & T_3 & T_4 & -T_1 
\end{pmatrix}.
\end{equation}
One can check that $(F_\Reg K_n)^{-1}F_\Reg$ coincides with the matrix $(AK_n)^{-1}A$ from \cref{ex:square1}.
\end{example}
 Substituting $t=t_0$ into~\eqref{eq:aff:curve} corresponds to substituting $T=\exp(it_0)$ into~\eqref{eq:gamma_fourier}, thus the $\C$-row span of $F_\Reg$ contains the subspace $\SpanC(\curve_\Reg)$. (The factor $\frac{(-1)^{n-j}}{(2i)^{n-1}}$ does not affect the row span of $F_\Reg$.) 

Given $0\leq t_1<t_2<\dots<t_n<\pi$ as in \cref{prop:rows_dist_pts}, let $A$ be the $n\times 2n$ matrix with rows  $\curve_\Reg(t_1),\curve_\Reg(t_2),\dots,\curve_\Reg(t_n)$. Then by~\eqref{eq:gamma_fourier}, we have $A=V\cdot F_\Reg$ where $V$ is a Vandermonde-like $n\times n$ matrix with $(m,j)$-th entry given by $\frac{(-1)^{n-j}}{(2i)^{n-1}}\exp((2j-n-1)it_m)$. One readily checks that this matrix is invertible, which implies that $A$ and $F_\Reg$ have the same $\C$-row span. We have shown the following result.

\begin{lemma}\label{lemma:span=F}
Let $\Reg$ be a valid region.
\begin{theoremlist}
\item\label{item:span(gamma)=F} The $\C$-row span of $F_\Reg$ equals $\SpanC(\curve_\Reg)$.
\item\label{item:dim<=n} The dimension of $\SpanC(\curve_\Reg)$ is at most $n$.
\item\label{item:span(A)=F} For any $0\leq t_1<t_2<\dots<t_n<\pi$, the $\C$-span of the vectors $\curve_\Reg(t_1),\curve_\Reg(t_2),\dots,\curve_\Reg(t_n)$ coincides with the $\C$-row span of the matrix $F_\Reg$.\qed
\end{theoremlist}
\end{lemma}
We will see later that whenever $\Reg$ is non-alternating, the matrix $F_\Reg$ has rank $n$, which will finish the proof of \cref{prop:span_dim_n}. For now, we observe that this is very easy to see in the generic case.

\begin{proposition}\label{prop:Span_dim=n}
Suppose that $\Reg$ is \emph{generic}, i.e., satisfies $\v_j\neq\v_k$ for all $j,k\in[2n]$ such that $j\neq k$. Then $\Span(\curve_\Reg)$ has dimension $n$.
\end{proposition}
\begin{proof}
By \cref{lemma:span=F}, it suffices to choose some values $0\leq t_1<t_2<\dots<t_n<\pi$ such that the matrix $A$ with rows $\curve_\Reg(t_1),\curve_\Reg(t_2),\dots,\curve_\Reg(t_n)$ has rank $n$. The rank of $A$ is invariant under permuting the rows, so it suffices to ensure that the values $t_1,t_2,\dots,t_n$ belong to $[0,\pi)$ and are pairwise distinct. Let us write $\J_1\sqcup\{1\}=\{j_1<j_2<\dots<j_n\}$, so $j_1=1$. For each $m\in[n]$, let $t_m\in[0,\pi)$ be equal to $\th_{j_m}$ modulo $\pi$. Because $\Reg$ is generic, the vectors $\v_j=\exp(2i\th_j)$ are pairwise distinct, thus $t_1,t_2,\dots,t_n$ are pairwise distinct. Now consider the matrix $A$ with rows $\curve_\Reg(t_1),\curve_\Reg(t_2),\dots,\curve_\Reg(t_n)$. Its submatrix with columns $\J_1\sqcup\{1\}$ is an $n\times n$ upper-triangular matrix with nonzero diagonal entries. Thus the matrix $A$ has rank~$n$.
\end{proof}

\section{Proof of the formula}
Let $\Reg=(\tau,\bth)$ be a non-alternating region. Our goal is to prove \cref{thm:main_span} by induction on $\xing(\tau)$. Along the way, we will show that the space $\Span(\curve_\Reg)$ has dimension $n$, which, in view of \cref{lemma:span=F}, will complete the proofs of \cref{prop:span_dim_n,prop:rows_dist_pts}.

The base case $\xing(\tau)=0$ corresponds to $\tau$ being a non-crossing matching. Let $j,k\in[2n]$ be such that $j<k=\tau(j)$. Then $\tau([j,k])=[j,k]$. The difference $k-j$ must be odd, and we set $\eps_{j,k}:=(-1)^{(k-j-1)/2}$. We see that $\J_j=\J_k$, so~\eqref{eq:curve} implies $\g_j(t)=\eps_{j,k}\g_k(t)$. Let us show by induction on $k-j$ that $\Span(\curve_\Reg)$ contains the vector $\bas_j+\eps_{j,k}\bas_k$, where $\bas_1,\bas_2,\dots,\bas_{2n}$ is the standard basis of $\R^{2n}$. When $k=j+1$, the non-alternating condition implies that $\g_j(\th_j),\g_k(\th_j)\neq0$ while $\g_m(\th_j)=0$ for all $m\in[2n]\setminus\{j,k\}$. Thus $\curve(\th_j)$ is proportional to $\bas_j+\eps_{j,k}\bas_k$, which shows the result for $k=j+1$. We may now modify the curve $\curve_\Reg$: let $\curve_\Reg^{[j,k]}$ be obtained from $\curve_\Reg$ by setting the $j$-th and $k=(j+1)$-th coordinates to zero and dividing all other coordinates by $\sin(t-\th_j)$ (which appears as a factor in $\g_m(t)$ for $m\in[2n]\setminus\{j,k\}$). Observe that $\Span(\curve_\Reg)$ coincides with $\Span(\curve_\Reg^{[j,k]})\oplus \Span(\bas_j+\eps_{j,k}\bas_k)$. We may now omit the pair $\{j,k\}$ from $\tau$ and repeat the process. Eventually, we decompose the space $\Span(\curve_\Reg)$ as a direct sum of spaces of the form $\Span(\bas_j+\eps_{j,k}\bas_k)$ over all pairs $j,k\in[2n]$ satisfying $j<k=\tau(j)$. Such vectors are linearly independent, which completes the base case for \cref{prop:span_dim_n,prop:rows_dist_pts}. It is also easy to see that the span of these vectors equals $\doublemap(M_\Reg)$, finishing the base case for \cref{thm:main_span}.
We note that the signs $\eps_{j,k}$ are uniquely determined by the property that the row vectors $(\bas_j+\eps_{j,\tau(j)}\bas_{\tau(j)})_{j\in\J_1\sqcup\{1\}}$ form an $n\times 2n$ matrix all of whose nonzero maximal minors have the same sign, cf.~\cite{GP}.

In order to complete the induction step, we need the following simple observation.
\begin{lemma}\label{lemma:gamma_ind}
Let $\Reg=(\tau,\bth)$ be a valid region and $k\in[2n]$ be a $\tau$-descent. Then for $\Reg':=\ska\Reg$ and all $t\in\R$, we have 
\begin{equation}\label{eq:gamma_ind}
\curve_\Reg(t)=\curve_{\Reg'}(t)\cdot \gth_k.
\end{equation}
\end{lemma}
\begin{proof}
 Denote $\curve_\Reg(t)=(\g_1(t),\g_2(t),\dots,\g_{2n}(t))$ and $\curve_{\Reg'}(t)=(\g'_1(t),\g'_2(t),\dots,\g'_{2n}(t))$. In order to treat the case $k=2n$ uniformly, we consider these sequences to be defined for all $k\in\Z$ via~\eqref{eq:aff:curve}. Let $a:=\taut(k)$ and $b:=\taut(k+1)$, thus we have $k<k+1<a<b$. Denote $\Reg'=(\tau',\bth')$, thus $\taut'(k)=b$, $\taut'(k+1)=a$, $\tht'_k=\tht_{k+1}$, and $\tht'_{k+1}=\tht_k$. 

Let $\Jt_m$ and $\Jt'_m$ be obtained from $\taut$ and $\taut'$ respectively by~\eqref{eq:aff:Jt}. By comparing $\Jt_m$ to $\Jt'_m$, we see that $\g_m(t)=\g'_m(t)$ for all $m\in\Z$ not equal to $k$ or $k+1$ modulo $2n$. A subtle point here is that for $a<m\leq b$, $\Jt_m$ contains $k+2n$ while $\Jt'_m$ contains $k+1+2n$ instead. However, we have $\tht_{k+2n}=\tht'_{k+1+2n}$, so we also get $\g_m(t)=\g'_m(t)$. It remains to consider the cases $m=k$ and $m=k+1$. We claim that there exists a function $\DUMMY$ such that 
\begin{align}
\label{eq:g_k:dummy}
\g_k(t)&=\sin(t-\tht_{k+1})\cdot \DUMMY, &\g_{k+1}(t)&=-\cos(t-\tht_k)\cdot \DUMMY;\\
\label{eq:g'_k:dummy}
\g'_k(t)&=\sin(t-\tht_{k})\cdot \DUMMY, &\g'_{k+1}(t)&=-\cos(t-\tht_{k+1})\cdot \DUMMY.
\end{align}
Indeed, we find that $k+1\in\Jt_k$, $a\in \Jt_{k+1}$, $k+1\in\Jt'_k$, and $b\in\Jt'_{k+1}$, with all other elements being common to $\Jt_k$, $\Jt_{k+1}$, $\Jt'_k$, and $\Jt'_{k+1}$. Using the fact that $\tht_a=\tht_k+\pi/2$ and $\tht'_b=\tht_{k+1}+\pi/2$, and letting $\DUMMY$ be the product of $\sin(t-\tht_m)$ for $m$ ranging over the common elements of these four sets, \eqref{eq:g_k:dummy}--\eqref{eq:g'_k:dummy} follow. We find
\begin{align}\label{eq:}
(\g'_k(t),\g'_{k+1}(t))\cdot \begin{pmatrix}
1/c_k & s_k/c_k\\
s_k/c_k & 1/c_k
\end{pmatrix}=\frac{\DUMMY}{c_k} \bigg(&\sin(t-\tht_k)-\sin(\tht_{k+1}-\tht_k)\cos(t-\tht_{k+1}),\\
&\sin(t-\tht_k)\sin(\tht_{k+1}-\tht_k)-\cos(t-\tht_{k+1})\bigg).
\end{align}
Writing 
\begin{equation}\label{eq:}
  \sin(t-\tht_k)=\sin \left((t-\tht_{k+1})+(\tht_{k+1}-\tht_k)\right),\quad\cos(t-\tht_{k+1})=\cos \left((t-\tht_k)-(\tht_{k+1}-\tht_k)\right)
\end{equation} 
and expanding using sum/difference formulas, we get
\begin{equation}\label{eq:}
(\g'_k(t),\g'_{k+1}(t))\cdot \begin{pmatrix}
1/c_k & s_k/c_k\\
s_k/c_k & 1/c_k
\end{pmatrix}=(\g_k(t),\g_{k+1}(t)).
\end{equation}
Since we have shown above that $\g'_m(t)=\g_m(t)$ for $m$ not equal to $k,k+1$ modulo $2n$, we are done with the proof of~\eqref{eq:gamma_ind}.
\end{proof}

Let us now complete the induction step in our proof. Let $\Reg=(\tau,\bth)$ and $\Reg'=(\tau',\bth')$ be as in \cref{lemma:gamma_ind}, thus $\xing(\tau)=\xing(\tau')+1$. Assume that the statement has been shown for $\curve_{\Reg'}(t)$. Thus $\Span(\curve_{\Reg'})$ has dimension $n$ and equals $\doublemap(M_{\Reg'})$. Comparing~\eqref{eq:gamma_ind} to~\eqref{eq:doublemap_ind}, we see that the same conclusions hold for $\Span(\curve_{\Reg})$. (Note that the matrix $\gth_k$ is always invertible.) This finishes the proof of \cref{prop:span_dim_n}, \cref{thm:main_span}, and \cref{prop:rows_dist_pts}. Recall that \cref{cor:main} has already been deduced from \cref{thm:main_span} in \cref{sec:double}.\qed

\section{Regular polygons}\label{sec:regular-polygons}
In this section, we prove \cref{thm:reg}. Before we proceed, we discuss some convergence results for the case of regular polygons approximating the unit disk.

\subsection{Asymptotics and convergence}
Let us discuss our formula~\eqref{eq:regular} in more detail. First, we clearly have $\<\sigma_1\sigma_1\>_{\Reg_n}=1$, and the next few values are
\begin{align}
  \<\sigma_1\sigma_2\>_{\Reg_n}&=\frac2n\cdot \frac1{\sin(\pi/2n)}-1, \\
  \<\sigma_1\sigma_3\>_{\Reg_n}&=\frac2n\left(\frac1{\sin(3\pi/2n)}-\frac1{\sin(\pi/2n)}\right) +1, \\
  \<\sigma_1\sigma_4\>_{\Reg_n}&=\frac2n\left(\frac1{\sin(5\pi/2n)}-\frac1{\sin(3\pi/2n)}+\frac1{\sin(\pi/2n)}\right)-1.
\end{align}
We stress that~\eqref{eq:regular} is valid not just asymptotically, but for all finite values of~$n$. Due to the explicit nature of this formula, computing the asymptotics becomes a straightforward exercise. For example, we have
\begin{equation}\label{eq:}
  \<\sigma_1\sigma_2\>_{\Reg_n}\to\frac4\pi-1,\quad \<\sigma_1\sigma_3\>_{\Reg_n}\to\frac4\pi \left(\frac13-1\right)+1,\quad \<\sigma_1\sigma_4\>_{\Reg_n}\to\frac4\pi \left(\frac15-\frac13+1\right)-1,\quad\dots
\end{equation}
as $n\to\infty$. By the Leibniz formula for $\pi$
\begin{equation}\label{eq:Leibniz}
\frac\pi4=1-\frac13+\frac15-\frac17+\frac19-\cdots,
\end{equation}
the above sequence $\lim_{n\to\infty}\<\sigma_1\sigma_k\>_{\Reg_n}$ tends to zero as $k$ increases. To get a nonzero limit, we need to multiply the correlations by~$n$.

\begin{corollary}\label{cor:reg_limit}
For $0<x<1$, we have
\begin{equation}\label{eq:reg_limit}
\lim_{n\to\infty} n\cdot \<\sigma_1\sigma_{\lfloor nx\rfloor}\>_{\Reg_n}=\frac1{\sin(\pi x)}.
\end{equation}
\end{corollary}
\noindent  Note that the scaling limit $\frac1{\sin(\pi x)}$ is consistent with the predictions of conformal field theory~\cite{BPZ2,BPZ1}. Specifically, applying a conformal map from the upper half plane to the unit disk, one can check that our scaling limit coincides with the function $\<\sigma(x_1)\sigma(x_2)\>_\Omega^{\operatorname{free}}$ defined in~\cite[Theorem~5]{Hongler}.

Despite the simplicity of these results, we have not been able to find any of them in the literature. The only result dealing with convergence of boundary spin correlations is due to Hongler~\cite[Theorem~5]{Hongler}. It applies to the case of the square lattice approximating a wide variety of regions, however, the unit disk is not among these regions, since its vertical and horizontal boundary parts are empty. 

\begin{remark}
Existing convergence results for Ising model correlations rely on the powerful technique of \emph{fermionic observables} developed in breakthrough work by Smirnov et al.~\cite{Smirnov,CS,HS,CDCHKS,CHI}. In particular, conformal invariance and universality for spin correlations in the interior of the region was shown  by Chelkak--Smirnov~\cite{CS}. They work in the context of arbitrary infinite rhombus tilings of the plane, however, a crucial assumption imposed e.g. in~\cite[Section~1.2]{CS} and~\cite[Section~1.2]{CS2} is that 
 \begin{equation}\label{eq:angles:CS}
\text{the angles of all rhombi are uniformly bounded away from $0$ and $\pi$}.
\end{equation}
It appears that this widely used assumption is completely necessary for the fermionic observables approach to apply. On the other hand, any rhombus tiling of a regular $2n$-gon $\Reg_n$ must contain a rhombus with angle~$\pi/n$. (In fact, for $n>2$, it contains exactly $n$ such rhombi.) In this sense, \cref{cor:reg_limit} is ``transversal'' to the convergence results obtained previously in the literature.
\end{remark}

\subsection{Proof of Theorem~\ref{thm:reg}}
Recall that $\Reg_n$ is a regular $2n$-gon, $M_{\Reg_n}$ is its correlation matrix, and $\double M_{\Reg_n}=(\double m_{j,k})$ is the associated ``doubled'' matrix defined in~\eqref{eq:double_signs}. By symmetry, it suffices to compute $\<\sigma_1\sigma_{j}\>_{\Reg_n}$ for $j\in[n]$. We are thus interested in the first row of $M_{\Reg_n}$, which is determined by the first row of $\double M_{\Reg_n}$. We have $\double m_{1,1}=\double m_{1,2}=1$. \Cref{thm:reg} is immediately implied by the following result.
\begin{lemma}
For $1<k<2n$, we have
\begin{equation}\label{eq:geom_prog}
\double m_{1,k}+\double m_{1,k+1}=
\begin{cases}
  0, &\text{if $k$ is odd;}\\
  (-1)^{\frac k2+1} \cdot \frac2{n\sin \left((k-1)\pi/2n\right)},&\text{if $k$ is even.}
\end{cases}
\end{equation}
\end{lemma}
\begin{example}
For the matrix $\double M=(AK_n)^{-1}A$ from \cref{ex:square1}, we have 
\begin{equation}\label{eq:}
  \double m_{1,2}+\double m_{1,3}=1+(\sqrt2-1)=\sqrt2=\frac1{\sin(\pi/4)} \quad\text{and}\quad\double m_{1,3}+\double m_{1,4}=(\sqrt2-1)+(1-\sqrt2)=0.
\end{equation}
\end{example}

\begin{proof}
Let $\zeta:=\exp(\pi i/2n)$. For $j\in[n]$, set $z_j:=\zeta^{2j-n-1}$. These are the $2n$-th roots of $(-1)^{n-1}$ with positive real part. Let $A=(a_{j,k})$ be an $n\times 2n$ matrix given by $a_{j,k}=z_j^{k-1}$. It is not hard to see from~\eqref{eq:fourier_coef} that for each $j\in[n]$, the $j$-th row of $A$ is a scalar multiple of the $j$-th row of $F_{\Reg_n}$. In particular, the row span of $A$ equals $\Span(\curve_{\Reg_n})=\doublemap(M_{\Reg_n})$, and we have $(AK_n)^{-1}A=\double M_{\Reg_n}$. We compute the first row of $(AK_n)^{-1}A$ explicitly.

First, we have $AK_n=DV$ where $D=\diag \left(\frac{1+z_1}2,\frac{1+z_2}2,\dots,\frac{1+z_n}2\right)$ is an $n\times n$ diagonal matrix and $V=(z_j^{2(k-1)})$ is an $n\times n$ Vandermonde matrix. In fact, $V$ is a variant of the \emph{discrete Fourier transform (DFT) matrix}, and similarly to the standard DFT matrix, it satisfies $V^{-1}=\frac1n V^*$, where $V^*$ is the conjugate transpose of $V$. We find that $(AK_n)^{-1}=V^{-1}D^{-1}$, thus its entries are given by
\begin{equation}\label{eq:}
  ((AK_n)^{-1})_{r,j}=\frac2{n(1+z_j)z_j^{2(r-1)}}\quad\text{for $r,j\in[n]$.}
\end{equation}
We can now multiply $(AK_n)^{-1}$ by $A$ and focus on the first row of the resulting matrix:
\begin{equation}\label{eq:}
\double m_{1,k}=\frac2n\sum_{j=1}^n \frac{z_j^{k-1}}{(1+z_j)}\quad\text{for $k\in[2n]$.}
\end{equation}
Computing $\double m_{1,k}+\double m_{1,k+1}$, we see that $(1+z_j)$ cancels out and we get
\begin{equation}\label{eq:}
\double m_{1,k}+\double m_{1,k+1}=\frac2n \sum_{j=1}^n z_j^{k-1}\quad\text{for $1\leq k<2n$.}
\end{equation}
This is a geometric progression; the result is easily seen to be given by~\eqref{eq:geom_prog}.
\end{proof}

\subsection{Proof of Corollary~\ref{cor:reg_limit}}
By symmetry, we only need to prove the result for $x\leq \frac12$, so fix $0<x\leq \frac12$. We treat the cases of $k$ odd and $k$ even differently, so fix $\eps\in\{0,1\}$ and let $k=k(n,x,\eps)$ be the integer closest to $nx$ such that $k\equiv \eps\pmod 2$.  Rewrite~\eqref{eq:regular}--\eqref{eq:Leibniz} as follows:
\begin{align}
\label{eq:limA}
(-1)^{k}\<\sigma_1\sigma_{k+1}\> &=
 1-\frac2{n\sin(\pi/2n)}+\frac2{n\sin(3\pi/2n)}-\dots+\frac{(-1)^k\cdot 2}{n\sin((2k-1)\pi/2n)};\\
\label{eq:limB}
\frac4\pi \sum_{m=k}^{\infty} \frac{(-1)^m}{2m+1}&=1-\frac4\pi+\frac4{3\pi}-\dots+\frac{(-1)^k\cdot 4}{(2k-1)\pi}.
\end{align}
We subtract~\eqref{eq:limB} from~\eqref{eq:limA}, multiply the result by $n$, and take the limit as $n\to\infty$. The left hand side becomes
\begin{equation}\label{eq:}
  (-1)^{\eps} \left(\lim_{n\to\infty} \<\sigma_1\sigma_{k}\>-\frac1{\pi x}\right).
\end{equation}
For the right hand side, we write:
\begin{equation}\label{eq:} \frac2{n\sin((r-2)\pi/2n)}-\frac2{n\sin(r\pi/2n)}-\frac4{(r-2)\pi}+\frac4{r\pi}=\frac{4\sin(\pi/2n)\cos((r-1)\pi/2n)}{n\sin(r\pi/2n)\sin((r-2)\pi/2n)}-\frac8{r(r-2)\pi}.
\end{equation}
This gives a Riemann sum approximating the integral of a continuous function:
\begin{equation}\label{eq:}
\int_0^x \left(  \frac{\pi\cos(\pi y)}{(\sin(\pi y))^2}-\frac1{\pi y^2} \right)\, dy=-\frac1{\sin(\pi x)}+\frac1{\pi x}.
\end{equation}

Since we are grouping the terms in the right hand sides of~\eqref{eq:limA}--\eqref{eq:limB} in pairs, there will be one extra term whenever $k$ is even. Multiplying this term by $n$, we get
\begin{equation}\label{eq:}
  n \left(\frac2{n\sin((2k-1)\pi/2n)}-\frac4{(2k-1)\pi}\right)\to 2 \left(\frac1{\sin(\pi x)}-\frac1{\pi x}\right) \qquad \text{as $n\to\infty$.}
\end{equation}
 Combining the pieces together, we find
\begin{equation}\label{eq:}
  (-1)^{\eps} \left(\lim_{n\to\infty} \<\sigma_1\sigma_{k}\>-\frac1{\pi x}\right) = -\frac1{\sin(\pi x)}+\frac1{\pi x} + 
  \begin{cases}
    0, &\text{if $\eps$ is odd}\\ 
    2 \left(\frac1{\sin(\pi x)}-\frac1{\pi x}\right),&\text{if $\eps$ is even.}
  \end{cases}
\end{equation}
This shows $\lim_{n\to\infty} \<\sigma_1\sigma_{k}\>=\frac1{\sin(\pi x)}$, regardless of the parity of $\eps$.\qed

\section{Alternating regions}\label{sec:alternating-regions}
So far we have only covered the case of non-alternating regions $\Reg$. The goal of this section is to extend our approach to arbitrary valid regions.

\def\L{L}
Let $\Reg=(\tau,\bth)$ be a valid region as defined in \cref{sec:match-regi-tilings}. For $\v\in\C$, we let 
\begin{equation}\label{eq:}
\L_\v:=\{k\in[2n]\mid \v_k=\pm \v\}.
\end{equation}
By~\eqref{eq:v=-v}, we have $\tau(\L_\v)=\L_\v$. The following is straightforward to check; see \figref{fig:alt}(right).
\begin{enumerate}[label=(\alph*)]
\item\label{item:NC} The restriction $\tau |_{\L_\v}$ of $\tau$ to $\L_\v$ is a non-crossing matching for each $\v\in\C$.
\item\label{item:alt_tau} If $\Reg$ is alternating then the indices $1\leq a<b<c<d\leq 2n$ in~\eqref{eq:alternating} may be chosen so that $b=\tau(a)$ and $d=\tau(c)$.
\end{enumerate}

In view of~\ref{item:alt_tau}, when $\Reg$ is an alternating region with $a,b,c,d$ as above, we see from~\eqref{eq:curve} that each $\g_k(t)$ becomes divisible by $\sin(t-\th_a)=\pm\sin(t-\th_c)$. Rescaling $\curve_\Reg(t)$ by $\frac1{\sin(t-\th_a)}$ does not change $\Span(\curve_\Reg)$, however, the coordinates of $\frac{\curve_\Reg(t)}{\sin(t-\th_a)}$ are trigonometric polynomials of degree $n-2$. Thus when $\Reg$ is alternating, $\Span(\curve_\Reg)$ has dimension strictly less than $n$, by the arguments in \cref{sec:fourier-transform}. In order to fix this, we need to take into account the ``derivatives'' of $\curve_\Reg(t)$.

For $k\in[2n]$, let $T_k:=\exp(i\th_k)$ so that $T_k^2=\v_k$. Let
\begin{equation}\label{eq:GG}
  \GG_\Reg(T)=(\gg_1(T),\gg_2(T),\dots,\gg_{2n}(T)), \quad\text{where}\quad \gg_k(T):=\prod_{j\in \Jt_k} \frac{T-\v_j}{T_j} \quad\text{for $k\in[2n]$.}
\end{equation}
Comparing with~\eqref{eq:gamma_fourier}, we find that $  \GG_\Reg(T^2)=(2iT)^{n-1}\cdot \curve_\Reg(t)$ for $T=\exp(it)$.
Thus the curves $\GG_\Reg$ and $\curve_\Reg$ have the same $\C$-span. For $k\in[2n]$, let
\begin{equation}\label{eq:}
\supp(k):=
\begin{cases}
  \{k,k+1,\dots,\tau(k)\}, &\text{if $k<\tau(k)$,}\\
  \{k,k+1,\dots,n,1,\dots,\tau(k)\},&\text{if $k>\tau(k)$.}
\end{cases}
\end{equation}
For a vector $x\in\C^{2n}$ and a set $S\subset[2n]$, we let $x|_S\in\C^{2n}$ be the vector obtained from $x$ by
\begin{equation}\label{eq:}
(x|_S)_k:=
\begin{cases}
  x_k, &\text{if $k\in S$,}\\
  0,&\text{if $k\notin S$.}
\end{cases}
\end{equation}
For $m\geq0$, denote by $\GG^\parr m_\Reg(T)$ the $m$-fold derivative of $\GG_\Reg(T)$:
\begin{equation}\label{eq:}
\GG^\parr m_\Reg(T):=\frac{d^m}{dT^m}\GG_\Reg(T).
\end{equation}
Finally, for $k\in[2n]$, let $m_k:=\#\{j\in\J_k\mid \v_j=\v_k\}$ be the degree with which $(T-\v_k)$ divides~$\gg_k(T)$. Denote
\begin{equation}\label{eq:}
  \u k:=\GG^\parr{m_k}_\Reg(\v_k)|_{\supp(k)} \quad\text{for $k\in[2n]$.}
\end{equation}
Thus $\u k\in\C^{2n}$ is obtained by (i) differentiating $\GG_\Reg(T)$ $m_k$ times, (ii) substituting $T=\v_k$, and (iii) sending all coordinates not in $\supp(k)$ to $0$.

\begin{theorem}\label{thm:arb_reg}
Let $\Reg$ be a valid region.  Then for each $k\in[2n]$, the vectors 
\begin{equation}\label{eq:u_j}
\{\u j\mid j\in \J_k\sqcup \{k\}\}
\end{equation}
form a basis of $\doublemap(M_\Reg)$.
\end{theorem}
\begin{proof}
As explained after \cref{dfn:non_alt}, any generic region is non-alternating. Clearly, any valid region $\Reg=(\tau,\bth)$ can be approximated by a sequence $(\RegN)_{N\geq0}$ of non-alternating valid regions with the same matching $\tau$: $\RegN=(\tau,\bthN)$. When we keep $\tau$ fixed, the matrix $M_\Reg$ depends continuously on $\bth$, thus $M_{\RegN}\to M_\Reg$ as $N\to\infty$. The map $\doublemap$ is also continuous (the topology on the set of $n$-dimensional subspaces of $\R^{2n}$ is inherited from the \emph{Grassmannian} $\Gr(n,2n)$). Thus $\doublemap(M_\Reg)$ consists of all vectors $x\in\R^{2n}$ obtained as limits of sequences of vectors $x^{(N)}\in\Span(\curve_{\RegN})$ as $N\to\infty$. Switching from working over $\R$ to working over $\C$ as in \cref{sec:real-compl-subsp}, it remains to show the following:
\begin{enumerate}
\item\label{item:lin_indep} the vectors in~\eqref{eq:u_j} are linearly independent;
\item\label{item:limit} for each $j\in [2n]$, $\u j$ can be obtained as a limit of a sequence $u^{(j,N)}\in \SpanC(\curve_{\RegN})$, where $(\RegN)_{N\geq0}$ is a sequence of generic valid regions approximating~$\Reg$.
\end{enumerate}
Let us first explain~\eqref{item:lin_indep}. Set $k=1$ and let $\{j_1<j_2<\dots<j_n\}=\J_1\sqcup\{1\}$. Let
 $A$ be the $n\times 2n$ matrix with rows $\u{j_1},\u{j_2},\dots,\u{j_n}$. Then the submatrix of $A$ with column set $\J_1\sqcup\{1\}$ is upper triangular with nonzero diagonal entries. Indeed, for $j\in \J_1\sqcup\{1\}$, the $j$-th coordinate of $\u j$ is nonzero since we have differentiated $\gg_j(T)$ by $T$ exactly $m_j$ times before substituting $T=\v_j$. On the other hand, for $j'\in\J_1$ such that $j'>j$, we have $j\notin\supp(j')$, so $(\u{j'})_{j}=0$. Thus the vectors in~\eqref{eq:u_j} are linearly independent for $k=1$. Because of the cyclic symmetry, the same holds for arbitrary $k\in[2n]$.

To show~\eqref{item:limit}, we need to introduce certain operators which are similar to the divided difference operators appearing in Schubert calculus. For a rational function $P(T)\in\C(T)$ defined at $a\in \C$, we set
\begin{equation}\label{eq:}
\dd{a} P(T):=\frac{P(T)-P(a)}{T-a}.
\end{equation}
Such operators commute: for $a\neq b$, we have $\dd{a} \circ\dd{b}=\dd{b}\circ\dd{a}$, and for a finite sequence $(a_1,a_2,\dots,a_m)$ of pairwise distinct complex numbers, we set 
\begin{equation}\label{eq:}
\dd{(a_1,a_2,\dots,a_m)} P(T):=\dd{a_1}\circ\dd{a_2}\circ\dots\circ \dd{a_m} P(T).
\end{equation}
If $P(T)\in\C[T]$ is a polynomial then so is $\dd{a} P(T)$. In this case, we claim
\begin{equation}\label{eq:multilimit}
  \lim_{a_1,\dots,a_m,a\to a_0} \left[ \left(\dd{(a_1,\dots,a_m)} P(T) \right) |_{T=a}\right]= \frac{1}{m!}\cdot P^\parr m (a_0) \quad\text{for $a_0\in\C$}.
\end{equation}
 Indeed, denoting the left hand side of~\eqref{eq:multilimit} by $H_m(P)$, one can show by induction on $m$ that for  $H_m(T^d)=0$ for $d<m$ and $H_m(T^d)={d\choose m} a_0^{d-m}$ for $d\geq m$. Extending this by linearity proves~\eqref{eq:multilimit} for polynomials in general.

An important property of $\supp(k)$ that we will need is that for all $j\in[2n]$, we have
\begin{equation}\label{eq:supp_zero}
j\notin \supp(k)\ \  \Longleftrightarrow\ \   k\in\J_j \ \ \Longrightarrow\ \  \gg_j(\v_k)=0, \quad\text{thus}\quad \GG_\Reg(\v_k)=\GG_\Reg(\v_k) |_{\supp(k)}.
\end{equation}

Let us go back to proving~\eqref{item:limit}. We let $\bth_\parr N=(\th_{\parr{N,1}},\th_{\parr{N,2}},\dots,\th_{\parr{N,2n}})$
 and $\v_\parr{N,k}:=\exp(2i\th_{\parr{N,k}})$. Thus for each $N$, the complex numbers $(\v_{\parr{N,k}})_{k\in[2n]}$ are pairwise distinct, and we have $\v_\parr{N,k}\to \v_k$ as $N\to\infty$. Fix $j\in[2n]$ and let
\begin{equation}\label{eq:}
  \Jp_j:=\{k\in \J_j\cap \supp(j)\mid \v_k=\v_j\}, \qquad\Jpp_j:=\{k\in \J_j\setminus \supp(j)\mid \v_k=\v_j\}. 
\end{equation}
Thus $\Jp_j\sqcup\Jpp_j\subset\L_{\v_j}$, and by our observation~\ref{item:NC} above, we have (cf. \cref{fig:NC}) 
\begin{figure}
\scalebox{0.95}{
\begin{tikzpicture}
\def\dotscl{0.3}
\def\labscl{1}
\def\RAD{2.1}
\draw(0,0) circle (\RAD);
\node[draw,circle, fill=black,scale=\dotscl] (Aj) at (0:\RAD) {};
\node[draw,circle, fill=black,scale=\dotscl] (Atj) at (165:\RAD) {};

\node[draw,circle, fill=black,scale=\dotscl] (Ak) at (30:\RAD) {};
\node[draw,circle, fill=black,scale=\dotscl] (Atk) at (100:\RAD) {};

\node[draw,circle, fill=black,scale=\dotscl] (Ak2) at (-150:\RAD) {};
\node[draw,circle, fill=black,scale=\dotscl] (Atk2) at (-50:\RAD) {};

\node[anchor=180,scale=\labscl] (Lj) at (Aj.0) {$j$};
\node[anchor=165+180,scale=\labscl] (Ltj) at (Atj.165) {$\tau(j)$};
\node[anchor=30+180,scale=\labscl] (Lk) at (Ak.30) {$k'$};
\node[anchor=100+180,scale=\labscl] (Ltk) at (Atk.100) {$\tau(k')$};
\node[anchor=-150+180,scale=\labscl] (Lk2) at (Ak2.-150) {$k''$};
\node[anchor=-50+180,scale=\labscl] (Ltk2) at (Atk2.-50) {$\tau(k'')$};

\draw(Aj) to[bend right=-10] (Atj);
\draw(Ak) to[bend right=-40] (Atk);
\draw(Ak2) to[bend right=-30] (Atk2);

\def\suppscl{0.7}
\node[scale=\suppscl] (Sj) at (0,0.3) {$\supp(j)$};
\node[scale=\suppscl] (Sk) at (0.8,1.4) {$\supp(k')$};
\node[scale=\suppscl] (Sk2) at (-0.4,-1.4) {$\supp(k'')$};
\end{tikzpicture}
}
  \caption{\label{fig:NC} Supports of elements in $\L_\v\cap J_j$; see~\eqref{eq:NC}.}
\end{figure}
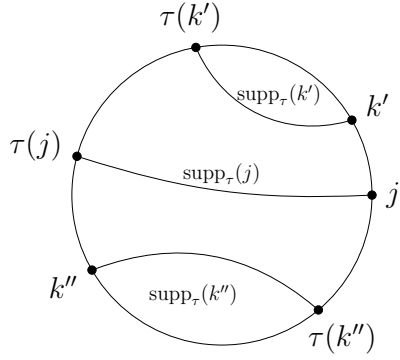
\begin{equation}\label{eq:NC}
\supp(k')\subset\supp(j) \quad\text{for $k'\in \Jp_j$,}\qquad \supp(k'')\cap\supp(j)=\emptyset \quad\text{for $k''\in \Jpp_j$.}
\end{equation}
 Let $m':=|\Jp_j|$, $\{k'_1,\dots,k'_{m'}\}:=\Jp_j$,  and $(a_1,\dots,a_{m'}):=(\v_{\parr{N,k'_1}},\dots,\v_{\parr{N,k'_{m'}}})$. Similarly, let $m'':=|\Jpp_j|$, $\{k''_1,\dots,k''_{m''}\}:=\Jpp_j$, and $(b_1,\dots,b_{m''}):=(\v_{\parr{N,k''_1}},\dots,\v_{\parr{N,k''_{m''}}})$. Thus $m_j=m'+m''$. Let 
\begin{equation}\label{eq:}
\GGBT:=\frac1{(T-b_1)\cdots (T-b_{m''})}\cdot  \GG_{\RegN}(T).
\end{equation}
By~\ref{item:NC}, we find that for each $k'\in\supp(j)$, $(\GG_{\RegN}(T))_{k'}$ is divisible by $(T-b_1)\cdots(T-b_{m''})$, and thus $(\GGBT)_{k'}$ is a genuine polynomial. For $k''\notin\supp(j)$, $(\GGBT)_{k''}$ is in general a rational function in $T$, but it is well defined at $T=a_1,\dots,T=a_{m'}$. Let $a:=\v_{\parr{N,j}}$. We extend the operator $P\mapsto \left(\dd{(a_1,a_2,\dots,a_{m'})}P\right)|_{T=a}$ to act coordinate-wise on tuples of rational functions and denote
\begin{equation}\label{eq:}
\w j:=\left(\dd{(a_1,a_2,\dots,a_{m'})} \GGBT \right)\big|_{T=a}.
\end{equation}
The vector  $\w j\in\C^{2n}$ is a linear combination of vectors $\{\GGB(v_{\parr{N,k}})\}_{k\in\Jp_j\sqcup\{j\}}$. First, each of these vectors belongs to $\SpanC(\GG_{\RegN})$, therefore so does $\w j$. Second, each of these vectors has $j''$-th coordinate equal to $0$ for $j''\notin\supp(j)$ by~\eqref{eq:supp_zero} and~\eqref{eq:NC}. Thus we have $\w j=\w j|_{\supp(j)}$. Third, as we mentioned above, the remaining entries of $\GGBT$ are polynomials in $T$, so~\eqref{eq:multilimit} applies to them. Letting $a_0:=\v_j$ and denoting $u^\parr{N,j}:=m'!\cdot \w j$, we apply~\eqref{eq:multilimit} to see that $u^\parr{N,j}\to \u j$ as $N\to\infty$. This shows~\eqref{item:limit}.
\end{proof}

\section{Concluding remarks}

\subsection{Relations to total positivity} \label{sec:relat-total-posit}
We briefly translate the objects studied above to the objects arising in total positivity. The matching $\tau$ is a \emph{decorated permutation} and the sets $(\J_j\sqcup\{j\})_{j\in[2n]}$ form a \emph{Grassmann necklace}; see~\cite{Pos}. The affine notation in \cref{sec:affine-notation} is mostly due to~\cite{KLS}: for instance, $\taut$ is a \emph{bounded affine permutation} in their language. The recurrence in \cref{sec:removing-crossing} is the ``BCFW bridge removal'' recurrence studied in~\cite{abcgpt,LamCDM} going back to~\cite{BCFW}. The image of the map $\doublemap$ is the \emph{totally nonnegative orthogonal Grassmannian} introduced in~\cite{HW,HWX}. The construction~\cite{GP} of this map was inspired by the results of Lis~\cite{Lis}.  The subspace $\doublemap(M_{\Reg_n})$ obtained in the case of a regular $2n$-gon is the unique cyclically symmetric point (cf. \cref{rmk:KW}) of the totally nonnegative Grassmannian studied in~\cite{GKL,KarpCS}.

\subsection{Explicit questions}
 Let $\Reg$ be a non-alternating region. Combining our results with the results of~\cite{GP}, it follows that $\Span(\curve_{\Reg})$ belongs to $\OGtnn(n,2n)$. In particular, for an $n\times 2n$ matrix $A$ whose row span is $\Span(\curve_\Reg)$ as in \cref{thm:main_span}, we have:
\begin{enumerate}
\item the nonzero maximal minors of $A$ all have the same sign;
\item the maximal minors of $A$ with complementary sets of columns are equal.
\end{enumerate}
The second condition can be equivalently restated as follows: for all $t,t'\in\R$, we have
\begin{equation}\label{eq:gamma_OG}
\g_1(t)\g_1(t')-\g_2(t)\g_2(t')+\dots +\g_{2n-1}(t)\g_{2n-1}(t')-\g_{2n}(t)\g_{2n}(t')=0.
\end{equation}
The curve $\curve_\Reg(t)$ is an explicit product of sines~\eqref{eq:curve}, which makes the above observations appear quite mysterious. We do not have a conceptual explanation for why they should be true. 
\begin{question}
Give an elementary proof that $\curve_\Reg(t)$ has the above properties.
\end{question}

\newcommand{\arxiv}[1]{\href{https://arxiv.org/abs/#1}{\textup{\texttt{arXiv:#1}}}}

\bibliographystyle{alpha}
\bibliography{ising_crit}

\end{document}